\documentclass{llncs}
\usepackage{llncsdoc}
\usepackage{algorithm}
\usepackage{algorithmicx}
\usepackage{algpseudocode}
\usepackage{graphicx}
\usepackage{epstopdf}
\usepackage{array}
\usepackage{amsmath}
\usepackage{amssymb}

\floatname{algorithm}{Algorithm}

\usepackage[lofdepth,lotdepth]{subfig}

\begin{document}
\title{Online Variant of Parcel Allocation in Last-mile Delivery}


\author{Yuan Liang \inst{1} \and Xiuguo Bao \inst{2}}
\institute{State Key Laboratory of Software Development Environment, School of Computer Science, Beihang University \and National Computer Network Emergency Response Technical Team/Coordination Center of China}

\maketitle

\begin{abstract}

We investigate the problem of last-mile delivery, where a large pool of citizen crowd-workers are hired to perform a variety of location-specific urban logistics parcel delivering tasks. Current approaches focus on offline scenarios, where all the spatio temporal information of parcels and workers are given. However, the offline scenarios can be impractical since parcels and workers appear dynamically in real applications, and their information is not known in advance. In this paper, in order to solve the shortcomings of the offline setting, we first formalize the online parcel allocation in last-mile delivery problem, in which all parcels were put in pop-stations in advance, while workers arrive dynamically. Then we propose an algorithm which provides theoretical guarantee for the parcel allocation in last-mile delivery. Finally, we verify the effectiveness and efficiency of the proposed method through extensive experiments on real and synthetic datasets.

\end{abstract}

\section{Introduction}
Last-mile delivery in urban logistics, where citizen volunteers are incentivized to deliver location-specific parcels, has recently attracted strong commercial interest \cite{2014traccs,2016_last-mile}. In real-life scenario, companies take the goods to the high-capacity warehouses (pop-stations), then the consumers will be notified and collect their parcels via mobile applications. If a parcel is not collected within 3 days, it is considered as a failed delivery. 

Therefore, in order to deal with unattended parcels, we utilize the power of crowd workers to enhance the last-mile delivery. In particular, crowd workers can take parcels from pop-stations to consumers. For example, the Yongjia of Beihang University is a pop-station, and many parcels will be collected from the pop-station, consumers will be noticed to collect their parcels by mobile applications. However, many parcels can not be delivered to consumers as consumers are not present when deliveries are made. This situation can be dealt with if there are some crowd-workers who can help consumers to collect parcels at pop-stations and deliver them to consumers. 
Recall the example of Beihang University, each student can become a crowd worker and collect parcels from Yongjia, and delivery them to consumers. As shown in Fig.~\ref{fig:parcels}, all parcels will be put in pop-stations, denoted by $p_1,p_2,...,p_n$, and all workers are denoted by $w_1,w_2,...,w_m$, parcels will be delivered to consumers by crowd-workers.

\begin{figure}[t]
\centering
\includegraphics[width = 0.8\textwidth]{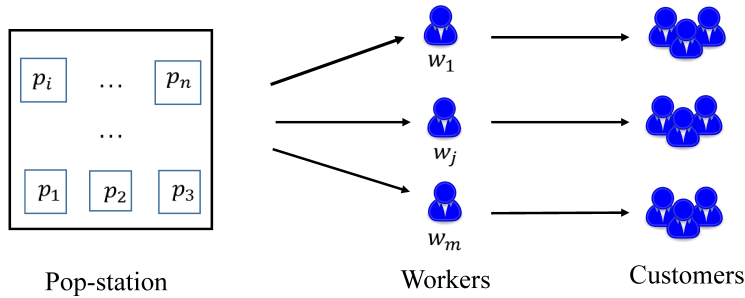}
\caption{The pattern of assign parcels to workers}
\label{fig:parcels}
\end{figure}

In this paper, we mainly discuss the problem of online parcel allocation in last-mile delivery, in which all parcels are put in the pop-station in advance, while crowd-workers use different time period to collect parcels, and each worker has different schedule that they will arrive dynamically. And crowd workers are rewarded with a certain amount of money 
according to their speed and the distance of the parcels. In addition, each worker has different working hours to collect different parcels within a day.  For example, Alice and Bob are crowd workers in a community, and there is a pop-station in their community, Alice use part-time to collect parcels from the pop-station to customers, and Bob is a full-time crowd worker. Therefore, crowdsourcing platform in order to pursuit the maximum benefit that they will give a reasonable parcel allocation. 
 
As introduced in \cite{2016_last-mile}, under offline scenarios, the parcel allocation problem in last-mile delivery can be solved by being reduced to the problem of min-cost flow model \cite{ahuja1993network}, where the source node $s$ is the only surplus and the sink node $t$ is the only demand node, and the remaining nodes represent parcels and crowd-workers, which serve as transshipment nodes. However, the reduction can not be performed under online scenarios, the offline solution becomes infeasible since the arrival orders of workers are unknown in dynamic environments. Therefore, it is necessary for us to propose an algorithm to solve online parcel allocation problem in last-mile delivery. To further illustrate this motivation, we go through a toy example as follows.

{\it Example 1}: Suppose in a social community, we have 8 parcels $p_1 \sim p_8$ and 4 workers $w_1 \sim w_4$ on a crowd-tasking platform, and each worker has a capacity, which is the maximum number of tasks that can be assigned to him/her. In this example, $w_1 \sim w_4$ have capacities of 2, 4, 3 and 2, respectively (in brackets). Table~\ref{table:utility} presents the utility (i.e. money) values between each pair of task and worker, which not only depends on the weight and distance of parcels also depends on the speed workers and the means of transportion of workers. Under the offline scenarios, we are able to derive all information of the parcels and workers, and the allocation is $\langle p_1, w_1 \rangle$, $\langle p_2, w_4 \rangle$,$\langle p_3, w_2 \rangle$, $\langle p_4, w_1 \rangle$, $\langle p_5, w_4 \rangle$, $\langle p_6, w_2 \rangle$, $\langle p_7, w_2 \rangle$, and $\langle p_8, w_3 \rangle$. However, in dynamic online scenarios, the information of workers  is not known in advance, and the allocation heavily depends on the order of crowd-workers. Suppose there are two orders of the crowd-workers, the first order is $w_2,w_3,w_4,w_1$, and the second order is $w_1,w_2,w_4,w_3$. For the first order, when $w_2$ arrives, $w_2$ will collect $p_3,p_4,p_6,p_7$; when $w_3$ arrives, $w_3$ will collect $p_2,p_5$; when $w_4$ arrives, $w_4$ will collect $p_4$; when $w_1$ arrives, $w_1$ will collect $p_1$. For the second order, when $w_1$ arrives, $w_1$ will collect $p_1,p_4$; when $w_2$ arrives, $w_2$ will collect $p_3,p_6,p_7$; when $w_4$ arrives, $w_4$ will collect $p_5$; finally, $w_3$ arrives, $w_3$ will collect $p_2$. Therefore, the online setting heavily depends on the order of crowd-workers.

\begin{table}[t]
\centering
\caption{Utility between parcels and workers}
\begin{tabular}{|p{1.08cm}<{\centering}|p{1.08cm}<{\centering}|p{1.08cm}<{\centering}|p{1.08cm}<{\centering}|p{1.08cm}<{\centering}|p{1.08cm}<{\centering}|p{1.08cm}<{\centering}|p{1.08cm}<{\centering}|p{1.08cm}<{\centering}|}
\hline
 & $p_1$ & $p_2$ & $p_3$ & $p_4$ & $p_5$ & $p_6$ & $p_7$ & $p_8$ \\
\hline
 $w_1 (2)$ & 0.9 & 0.4 & 0.5 & 0.9 & 0.4 & 0.8 & 0.3 & 0.9  \\
 \hline
 $w_2 (4)$ & 0.2 & 0.2 & 0.6 & 0.3 & 0.2 & 0.9 & 0.8 & 0.3  \\
 \hline
 $w_3 (3)$ & 0.4 & 0.5 & 0.2 & 0.4 & 0.7 & 0.2 & 0.2 & 0.7  \\
 \hline
 $w_4 (2)$ & 0.3 & 0.6 & 0.4 & 0.6 & 0.9 & 0.4 & 0.9 & 0.2  \\
 \hline 
 \end{tabular}
\label{table:utility}
\end{table}

To the best of our knowledge, this is the first work that studies the online scenario in last-mile delivery problem. It is therefore crucial to design an efficient and effective on-line algorithm dedicated for the last-mile delivery. In this paper, we make the following contributions:

a). We identify a new online scenario of parcel allocation in last-mile delivery model.

b). We propose an algorithm with slight power, whose competitive ratio is $\frac{1}{2(1+\lfloor \log(\mu)\rfloor)}$, where $\mu \geq 1$.

c). We conduct extensive experiments on real and synthetic datasets to evaluate the efficiency and effectiveness of our proposed algorithms.

The rest of the paper is organized as follows. We present the problem definition in Section 2. In Section 3, we review related works. The algorithms of our problem and the theoretical analysis is proposed in Section 4. We conduct experiments in Section 5 to evaluate the performance of our proposed solutions. The conclusion of our work is given in Section 6.

\section{Problem Statement}

We first introduce several concepts and then formally define the dynamic online scenario in last-mile delivery.

{\it Definition 1 (Parcel)}: All parcels are located in pop-stations, and there are $n$ parcels, denoted by $P = \{p_1,p_2,...,p_n\}$. Each parcel has a utility for each worker, denoted by $p_{ij}$. 

{\it Definition 2 (Worker)}: A crowd worker ("worker" for short), denoted by $w = <T_j, c_j, t_{ij}>$, where $T_j$ represents the time of worker $j$ collect parcels within a day. In addition, capacity $c_j$ is the maximum number of parcels that worker $j$ intends to collect, and $t_{ij}$ denotes the time of worker $j$ delivery parcel $i$. And there are $m$ workers, denoted by $W = \{w_1,w_2,...,w_m\}$. 

We then define the utility value that parcels are allocated workers as follows.

{\it Definition 3 (Utility value)}: The utility value that a worker $j$ perform a parcel $i$ is measured by $U(i,j) = p_{ij} x_{ij}$.

where $p_{ij}$ denotes the utility of parcel $i$ is collected by worker $j$. Let $x_{ij} = 1$ denote parcel $i$ is assigned to worker $j$ and $x_{ij} = 0$ for the opposite case. Our goal is to maximize that

\begin{equation}
\sum_i \sum_j p_{ij} x_{ij} 
\end{equation}
subject to
\begin{eqnarray}
\sum_j x_{ij} & \leq & 1   \\
\sum_i x_{ij} & \leq & c_j  \\
\sum_i t_{ij} x_{ij} & \leq & T_j
\end{eqnarray}

Constraint (2) means that each parcel must be assigned to no more than one worker. Constraint (3) restricts a worker has a capacity, which is the maximum number of parcels that can be assigned to him/her. Constraint (4) means that each worker has different working hours in a day (i.e. some workers are part-time, and some workers are full-time), where $t_{ij}$ denotes the time of worker $j$ collect parcel $i$, and $T_j$ represents the working hours of worker $j$ within a day. 

The online algorithm has another constraint that once a parcel $p$ is allocated to a worker $w$, the allocation of $(p,w)$ cannot be changed. And the performance of online algorithm is usually compared with the optimal allocation of the offline scenario and heavily depends on the arriving orders of crowd-workers. Moreover, we evaluate the online algorithm using the notion of competitive ratio, which is a lower bound on the ratio between the utility of the algorithm and the utility of the optimal offline algorithm (all information of all parcels and workers are known in advance). For example, a competitive ratio of 1/2 would imply that an algorithm always achieves a utility that is at least half as good as optimal. We say an algorithm is $\alpha-competitive$ if its competitive ratio is at least $\alpha$.

The optimal solution of the offline problem is introduced as follows.

We introduce an optimal offline parcel allocation in last-mile delivery, we can reduce our offline problem to the min-cost flow problem. We first construct a flow network $G = (V,E)$, where $V = P \cup W \cup \{s,t\}$, $s$ is a source node and $t$ is a sink node. In flow network, there are three types of arcs, denoted by $s \to p_i$, $p_i \to w_j$ and $w_j \to t$ respectively, and the network is shown in Fig.~\ref{fig:flow_net}. The first type arcs $s \to p_i$ represents the source nodes to parcels, with capacity 1 and the cost of this type arcs is 0; the second type arcs $p_i \to w_j$ is from parcels to workers, with capacity 1 and the cost of this type arcs is $U'$, where $U'(p,w) = \sum_i \sum_j (\rho - p_{ij})$, $\rho = \max p_{ij} + 1$; The third type arcs $w_j \to t$ is from workers to sink node, with capacity $c_j$ is the maximum number parcels of worker $w_j$ can collect. Then, we use existing flow algorithms to obtain the optimal offline parcel allocation in last-mile delivery, e.g., simplex algorithm \cite{kiraly2012}, to calculate the optimal utility value.

\begin{figure}[t]
\centering
\includegraphics[width = 0.8\textwidth]{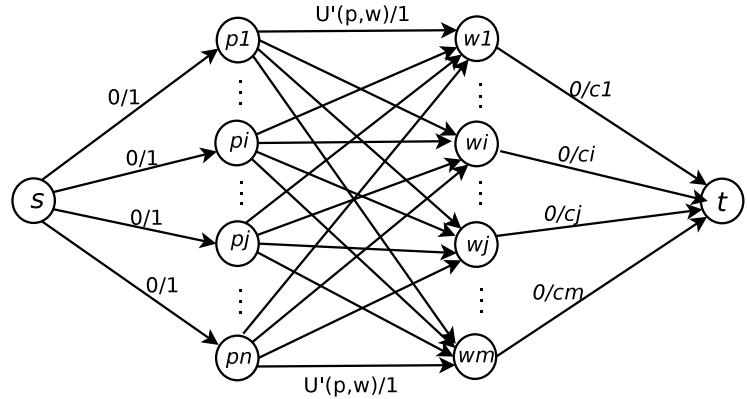}
\caption{Min-cost flow model}
\label{fig:flow_net}
\end{figure}

\section{Related Work}
In this section, we review related works from two categories, spatial crowdsourcing and last-mile delivery.

\subsection{Mobile crowdsourcing}

In recent years, with the development of mobile Internet and distributed systems \cite{tongwww16tracking,tongjos17}, more and more applications of mobile crowdsourcing are emerging, e.g. Uber, Gigwalk, etc. Particularly, task assignment is one of the major topics on mobile crowdsourcing. \cite{zheng2015qasca,fan2015icrowd} introduced task assignment problem in offline setting by learning the quality of crowd workers. 

In addition, there have been several researches on online allocation, such as \cite{2012_online_matching,ICDE16_Tong,tongpvldb16}. \cite{ICDE16_Tong} propose a general bipartite-matching-based framework to address the dynamic task allocation in online spatial crowdsourcing platforms. Moreover, \cite{ho2012online} studies online task assignment, where tasks arrival while workers are not dynamic. And\cite{zhangpvldb14} introduced that minimize the total waiting time before the crowd workers arrives at a specific location of tasks.　Although the aforementioned works study the various problem on online allocation, most of them do not address the last-mile delivery.

\subsection{Last-mile delivery}

In recent years, there have been several researches in the last-mile delivery\cite{2016_last-mile,2014traccs}. \cite{2016_last-mile} introduced the planning and scheduling parcels so as to minimize the additional travelling cost. They propose an effective large-scale mobile crowd-tasking model in which a number workers deliver a variety of location-specific urban logistics parcels. However, they do not address the parcel allocation problem for workers that arrive dynamically, and are thus impractical for dynamic environments.

In addition, there are many approaches which focus on a particular class of the workers' utility. For example, \cite{shahabi2012} introduced a centralized allocation algorithm that maximize allocated tasks, while satisfying a set of constraints. \cite{he2014toward} consider that tasks with different requirements of quality of sensing are typically associated with specific locations and mobile users are constrained by time budgets. Moreover,  \cite{2014traccs} study the task allocation as an optimization problem with the objective function that maximizing the total utility from all assigned tasks. Although these works study the task allocation problem on mobile crowdsourcing, they rarely consider the challenge of the last-mile delivery and dynamic scenario.

\section{Solution}

In this section, we propose a baseline algorithm and a primal-dual algorithm to solve the parcel allocation in last-mile delivery. In our problem, all parcels were taken to pop-station, when a crowd-worker $j$ arrives, s/he can carry at most a set of parcels $c_j$. And when a worker arrives, parcels will be allocated to the worker immediately, and the allocation can not be changed. 

\subsection{Baseline Algorithm }

In this subsection, we present a greedy as baseline algorithm. The main idea of greedy is to set all parcels in pop-station, and crowd-workers arrive dy dynamically. Note that when a crowd-worker $j$ arrives, s/he would carry at most $c_j$ parcels once. In particular, each crowd-worker will select the parcels with the highest utility that satisfies all constraints.

More specifically, let $m$ is the number of workers, $n$ is the number of parcels, and $M = \{(i,j)|i \in P, j \in W, i \geq j\}$ is the allocation. When a crowd-worker $j$ arrives, s/he has a capacity constraint $c_j$. Let $S(j)$ represents the set of parcels of worker $j$ will collect. When worker arrives, s/he will repeat this process until there are no any parcels and thereby completes the allocation. 

\begin{algorithm}[!t]
\caption{Greedy}
\begin{algorithmic}[1]
\Require $P, W, U(,.,)$
\Ensure A feasible allocation $M$
\State $S(w)\gets \emptyset$ for all $w$.
\For {each parcel $p_i$}
\State buile a list $q_i$ sorted in ascending order of $U(p_i,w_j)$ for all workers
\EndFor
\For {each new arrival worker $w$}
\State Let $S(p)$ denote the set of parcels are not allocated
\If {$S(p) \neq \emptyset$}
\State $S(w) \gets $the worker collect a set of parcels with highest utility that satisfies all constraints
\State $S(p) \gets S(p) - S(w)$
\Else
\State break;
\EndIf
\EndFor
\State \Return{the final allocation $M$ and the total utility $U$.} 
\end{algorithmic}
\label{algo:greedy}
\end{algorithm}

Algorithm~\ref{algo:greedy} illustrates the procedure. Line 1 initializes $S(w)$ for each worker. Here, $S(w)$ denotes worker $w$ collect a set of parcels at most $c_w$. In lines 2-4, we build a list $q_i$ sorted in ascending order of $U(p_i,w_j)$ for all workers such that when a worker arrives, parcels will quickly be allocated to the worker. In lines 5-13, we iteratively process each new arrival worker. Particularly, we adopt a greedy strategy on the arrival worker. In addition, let $S(p)$ denotes the set of parcels are not allocated in line 6. In lines 7-12, if $S(p) \neq \emptyset$, the arrival worker will collect a set of parcels with highest utility that satisfies all constraints. The capacity of arrival worker $j$ at most $c_j$. Finally, we will obtain the maximum utility $U$ and the final allocation $M$ accordingly in line 14.

\begin{example}
Here is the process of running Greedy algorithm on Example 1. First, we set the utility of workers and parcels in TABLE~\ref{table:utility}. Suppose the order of workers is $w_2, w_4, w_3, w_1$, when worker $w_2$ arrives, s/he will collect parcels $p_3, p_4, p_6, p_7$ that satisfies her/his capacity and obtain the maximum utility. When worker $w_4$ arrives, s/he will collect $p_2, p_5$, the utility of worker $w_4$ is $0.6 + 0.9 = 1.5$. When worker $w_3$ arrives, s/he will collect $p_8$ that satisfies $w_3$ constraints. When $w_1$ arrives, there is only a parcel $p_1$ is not collected, and $w_1$ will collect $p_1$, the utility of $w_1,p_1$ is 0.9. Finally, the online parcel allocation is $\langle p_1,w_1 \rangle, \langle p_2,w_4 \rangle, \langle p_3,w_2 \rangle, \langle p_4,w_2 \rangle, \langle p_5,w_4 \rangle, \langle p_6,w_2 \rangle, \langle p_7,w_2 \rangle, \langle p_8,w_3 \rangle$. 

\end{example}

{\bf Complexity analysis}. In lines 2-4, we build an inverted list for each parcel $p_i$, the elements in the list are in form of $p(i,j)$, and the time complexity of Greedy is $O(|W||P|)$. For each of the new arrival worker, the time complexity of Greedy is $O(\log(|W||P|))$. Therefore, the time complexity of Greedy is $O(|W||P|+\log(|W||P|))$.

\subsection{Primal-dual Algorithm}

Inspired by online primal-dual algorithm \cite{primal-dual2009,2005online,2013unified}, we propose a primal-dual algorithm to solve parcel allocation problem. We first give the duality of the problem, and the standard primal-dual linear programming formulation of our problem are shown as follows.

\begin{equation}
maximize  \sum_i \sum_j p_{ij}  x_{ij} 
\end{equation}
subject to
\begin{eqnarray}
\sum_j x_{ij}  & \leq & 1, j \in W     \\
\sum_i x_{ij} & \leq &  c_j, i \in P   \\
\sum_i t_{ij} x_{ij}  & \leq & T_j, (i,j) \in M
\end{eqnarray}

According to \cite{primal-dual}, and the derivations of its dual are shown as follows.

{\bf{Step 1}}. Rewrite the objective (5) as a minimization.

\begin{equation}
minimize  -\sum_i \sum_j p_{ij} x_{ij} 
\end{equation}

{\bf{Step 2}}. Rewrite each inequality constraint as a "less than or equal", and rearrange each constraint so that the right-hand side is 0.

After this step our linear program now looks as follows.

\begin{equation}
minimize  -\sum_i \sum_j p_{ij} x_{ij} 
\end{equation}
subject to
\begin{eqnarray}
\sum_j x_{ij}  - 1 & \leq & 0, j \in W     \\
\sum_i x_{ij} -  c_j & \leq & 0, i \in P   \\
\sum_i t_{ij} x_{ij}  - T_j & \leq & 0, (i,j) \in M
\end{eqnarray}

{\bf{Step 3}}. Define a non-negative dual variable for each inequality constraint. To constraint (12) and (13) we associate variable $\alpha_i \geq 0$, and to constraint (11), we associate variable $\beta_j \geq 0$.

{\bf{Step 4}}. For each constraint, eliminate the constraint and add $\alpha_i$, $\beta_j$ to the objective. Concretely, for the constraints (12) and (13) we would remove it and add the following term to the objective.

\begin{equation}
\alpha_i(\sum_i x_{ij} -  c_j + \sum_i  t_{ij} x_{ij} - T_j)
\end{equation}

If we do this for each constraint, and maximize the result over the dual variables, we get

$$max_{\alpha_i \geq 0, \beta_j \geq 0} min_{x_{i,j}\geq 0, p_{i,j} \geq 0} -\sum_i \sum_j p_{ij} x_{ij} $$

\begin{eqnarray}
& + & \alpha_i ( \sum_i t_{ij} x_{ij}   - T_j + \sum_i x_{ij} -  c_j)  +  \beta_j ( \sum_j x_{ij}  - 1 )
\end{eqnarray}

{\bf{Step 5}}. Rewrite the objective. And we get

$$max_{\alpha_i \geq 0, \beta_j \geq 0} min_{x_{i,j}\geq 0, p_{i,j} \geq 0} -\alpha_i(T_j + c_j) - \beta_j $$

\begin{eqnarray}
& + & \alpha_i ( \sum_i t_{ij} x_{ij}  + \sum_i x_{ij})    +  \beta_j \sum_j x_{ij}   -  \sum_i \sum_j p_{ij} x_{ij}
\end{eqnarray}

{\bf{Step 6}}. Remove $\alpha_i$ and $\beta_j$ with constraints.

{\bf{Step 7}}. Finally, rewrite the result of the step 6 as a minimization. And the dual is shown as follows.

\begin{equation}
minimize  \sum_i \alpha_i ( T_j + c_j ) + \sum_j \beta_j 
\end{equation}
subject to
\begin{eqnarray}
\alpha_i ( T_j + c_j ) + \beta_j  & \geq &   p_{ij}, (i,j) \in M  \nonumber \\
\alpha_i  & \geq &  0, i \in P   \nonumber \\
\beta_j & \geq &  0, j \in W
\end{eqnarray}

\begin{algorithm}[t]
\caption{Primal-dual}
\begin{algorithmic}[1]
\State initially $\forall x_{ij} \gets 0 $
\For {Upon arrival of a new worker $j$ allocate a set of parcels to the worker $j$ that maximizes $\sum_i p_{ij}x_{ij}$}
\If {$x_{ij} = 0$}
\State Calculate the utility of worker $j$ and its remaining time $T_j - \sum_i t_{ij}x_{ij}$ and set $x_{ij} = 1$
\State $\alpha_i \gets t_{ij}(1-x_{ij})$
\State $\beta_j \gets p_{ij} (1 - x_{ij}) + x_{ij}$
\Else
\State break;
\EndIf
\EndFor
\State \Return{the final allocation $M$ and the total utility $U$.} 
\end{algorithmic}
\label{algo:primal-dual}
\end{algorithm}

First, we tempting to design a online primal-dual algorithm (i.e., in Algorithm~\ref{algo:primal-dual}) that achieves a constant worst-case competitive ratio, then we found that no algorithm can be constant-competitive via the following theorem.

\begin{theorem}
Suppose $\mu \geq 1$, consider instances of our problem when $\forall (i,j) \in M: t_{ij} \leq T_j \leq \mu t_{ij}$. Then worst-case competitive-ratio of last-mile delivery on these instances is $\frac{1}{2(1+\lfloor \log(\mu)\rfloor)}$.
\end{theorem}

\begin{proof}
Suppose there are a set of parcels in the pop-station, and let $m = 2^k -1$ parcels, for $k \in \mathbb{N}$. The parcels partitioned into $k$ sets, denoted $K = \{0,1,...,k-1\}$, and the set $t$ has $s^t$ parcels. When a worker $j$ arrives, s/he will collect $c_j$ parcels. When algorithm runs on this input, let 

\begin{equation}
T_j =  \sum_{i \in K} t_{ij}x_{ij} 
\end{equation}

Note that due to feasibility of the algorithm, $t_{ij} \leq T_j$. Now, run the algorithm on the input with the parameter $0 \leq s \leq k-1$. Obviously, the output of the algorithm is feasible. To show the competitive ratio, partition the allocation set $M$ into $M_0,M_1,...,M_{\log(\mu)}$ such that 

\begin{equation}
(i,j) \in M_s   \Longleftrightarrow \lfloor \log(\frac{T_j}{t_{ij}}) \rfloor= s
\end{equation}

Furthermore, suppose $\{ x_{i,j} \}$ is the algorithm's allocation and $\{x_{ij}^*\}$ is the optimal offline allocation of the last-mile delivery. Hence, $OPT = \sum_{(ij) \in M} p_{ij}x_{ij}^*$. Now, from $M_s$, we have 

\begin{equation}
(i,j) \in M_s   \Longleftrightarrow s \leq \log(\frac{T_j}{t_{ij}}) < s+1 \Longleftrightarrow 2^s \leq \frac{T_j}{t_{ij}} < 2 ^{s+1}
\end{equation}

Let $\widetilde{T_j} \triangleq 2^s p_{ij}$, and $\widetilde{T_j} \geq \frac{1}{2}T_j$. And we have

\begin{eqnarray}
E\{\sum_{(i,j)\in M} p_{ij} x_{ij} \vert s \}  \geq  E\{ 2^s \sum_{(i,j)\in M}p_{ij}x_{ij}\}   \nonumber \\
\geq  E\{ 2^s\sum_{(i,j)\in M}p_{ij}x^*_{ij} |s \}   = \sum_{(i,j)\in M_s} E\{ \widetilde{T}_j x^*_{ij} |s \}
\end{eqnarray}

To complete the proof of Theorem 1, note that $OPT = \sum_i \sum_j p_{ij} x_{ij}^* $. Here we get

\begin{eqnarray}
E\{\sum_{(i,j)\in M} p_{ij} x_{ij} \}  =   \sum_s Pr\{s\}E\{\sum_{(i,j)\in M}p_{ij}x_{ij} \vert s\} \nonumber \\
 \geq  \frac{1}{1+\lfloor \log(\mu) \rfloor}\sum_s\frac{1}{2}\sum_{(i,j)\in M_s} p_{ij} x^*_{ij} =  \frac{1}{2(1+\lfloor \log(\mu)\rfloor)}OPT
\end{eqnarray}
\end{proof} 

Therefore, we found that no algorithm can be constant-competitive, and we obtain logarithmically competitive ratio on a certain parameter of our problem through primal-dual algorithm.

\section{Experimental Evaluation}

\subsection{Experiment Setup}

In this subsection, we evaluate our proposed algorithms. We use both real and synthetic datasets for experiments.

{\bf Datasets.} We use the GeoLife dataset from \cite{zheng2009mining} as the real dataset. In the GeoLife dataset, there are 182 users, and we generate the number of parcels {200,400,600,800,1000}. For the synthetic data, we generate the working hours of workers following normal distributions, the capacity of workers following uniform distributions.  And the utility (i.e., money) of parcels for each worker are generated following (10,20) randomly. The statistics and configurations of the synthetic data are illustrated in Table~\ref{table.synthetic}. Default settings are denoted in bold font.

\doublerulesep 0.1pt
\begin{table}[t]
\begin{footnotesize}
\begin{center}
\caption{Synthetic dataset}
\begin{tabular}{|p{5cm}<{\centering}|p{5cm}<{\centering}|}
\hline\hline
Factor & setting \\
\hline
$|P|$ & 100,{\bf 200}, 300, 400, 500 \\
\hline
$|W|$ & 20, {\bf 40}, 60, 80, 100 \\
\hline
$|T|$ & ${\mu=5, \sigma = \{2,3,4,5,6\}}$  $ {\mu = \{2,3,4,5,6\}, \sigma = 5}$ \\
\hline
$|c|$ & [1,6] \\
\hline
Scalability ($|P|$) & $10k, 20k, 30k, 40k, 50k$ \\
\hline\hline
\end{tabular}
\label{table.synthetic}
\end{center}
\end{footnotesize}
\end{table}

We evaluated our algorithms in terms of allocation utility, running time and memory cost and study the effects of varying the parameters on the performance of the algorithms. The offline scenario is solved by simplex algorithm through min-cost flow model. And the synthetic datasets are created in Python, and all algorithms are implemented in C++ and executed under the Linux Ubuntu operating system. The experiments are conducted on a computer with an Intel Xeon E5620 with a 2.40 GHz 16-core CPU and 12 GB of memory.

\subsection{Experiment Results}
In this section, we evaluate the proposed algorithms in terms of allocation utility, running time and memory cost. We test the performances of the proposed algorithms by varying the parameters as follows: the size of $W$, the size of $P$, the capacity of worker $c$, and the working hours of workers $T$.

\begin{figure*}[t]
\centering

\subfloat[\small{Utility of varying $|W|$}\vspace{-2ex}]{
\includegraphics[scale=0.19]{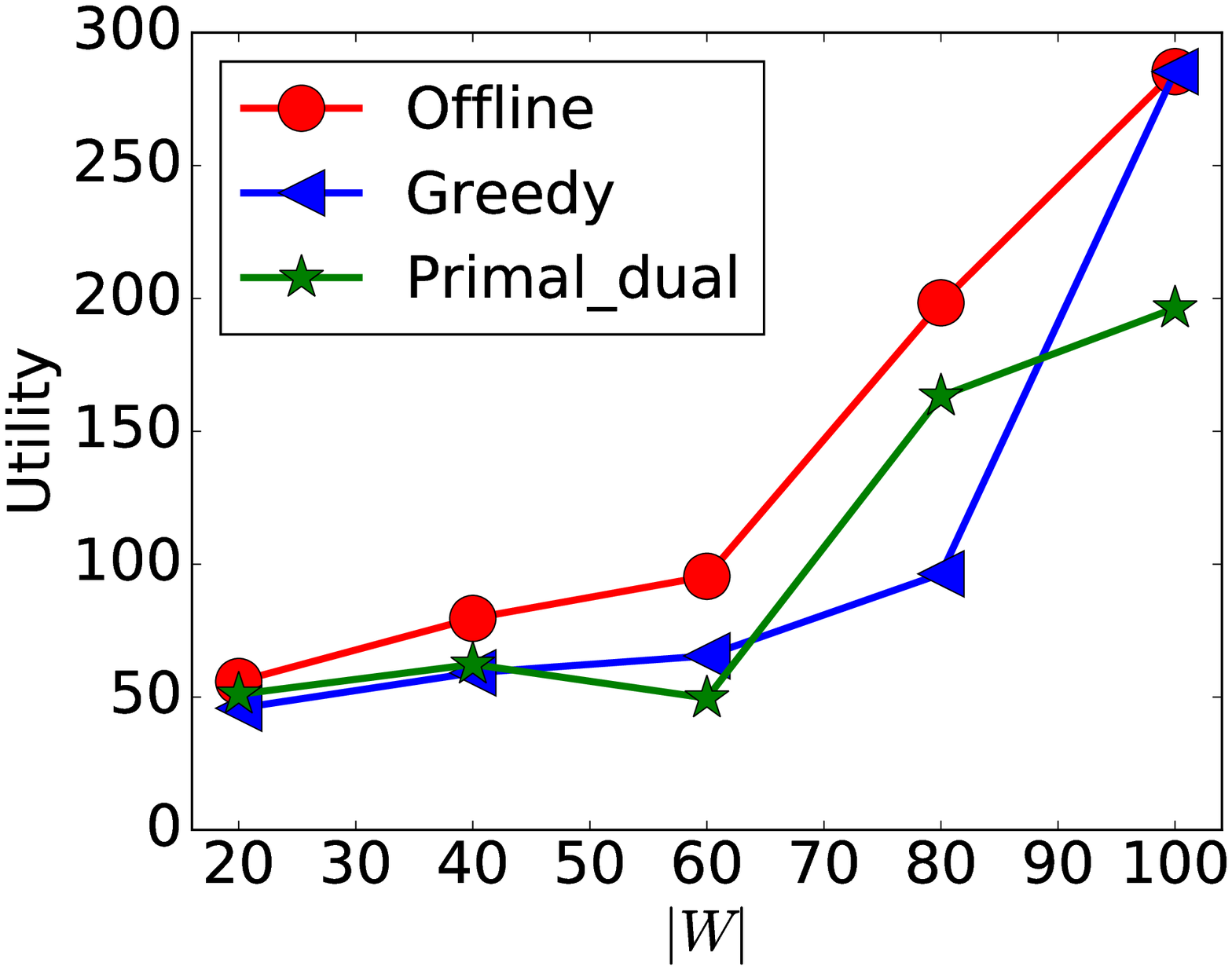}
\label{fig:w_cost}
}
~~
\subfloat[\small{Time of varying $|W|$}\vspace{-2ex}]{
\includegraphics[scale=0.19]{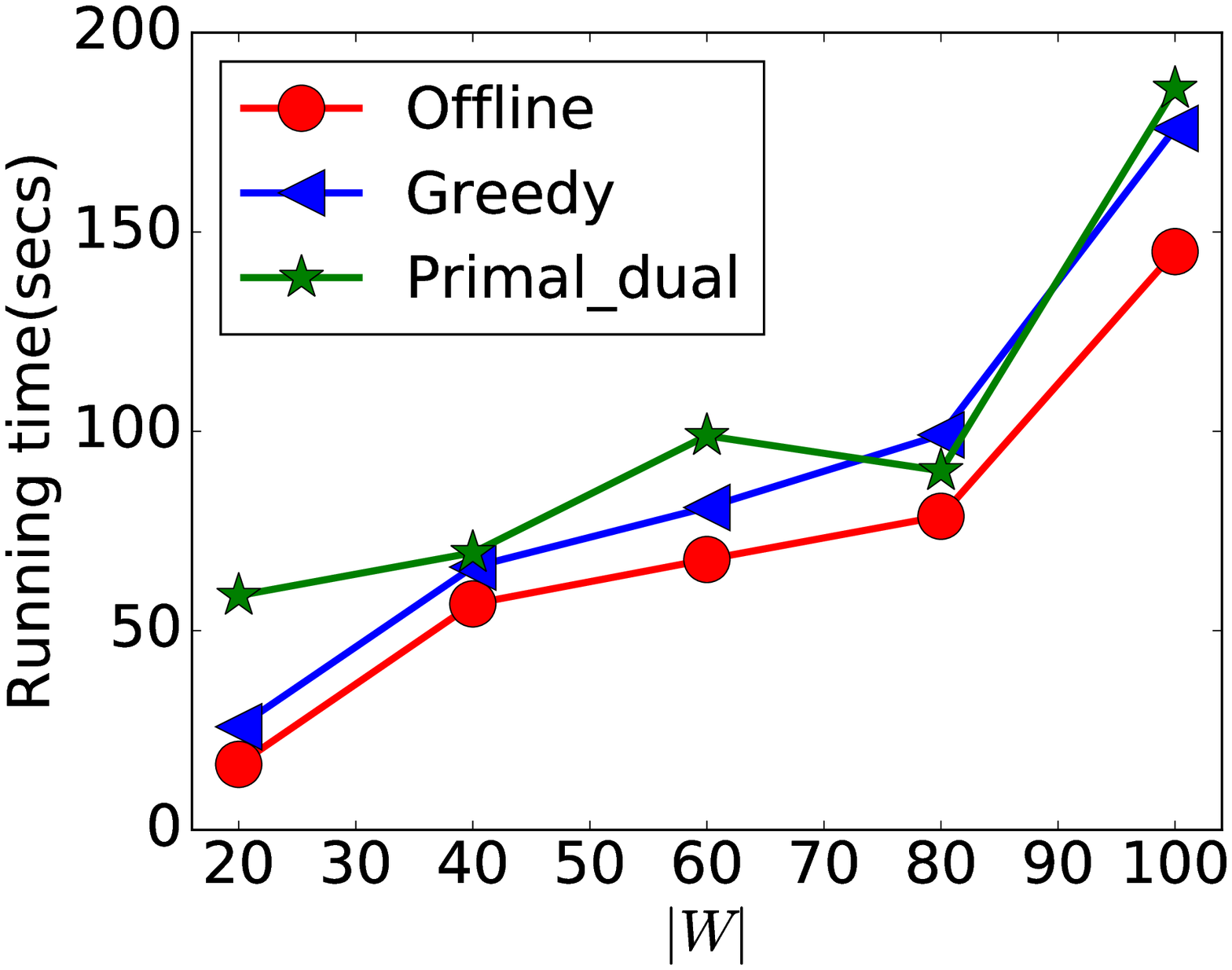}
\label{fig:w_time}
}
~~
\subfloat[\small{Memory of varying $|W|$}\vspace{-2ex}]{
\includegraphics[scale=0.19]{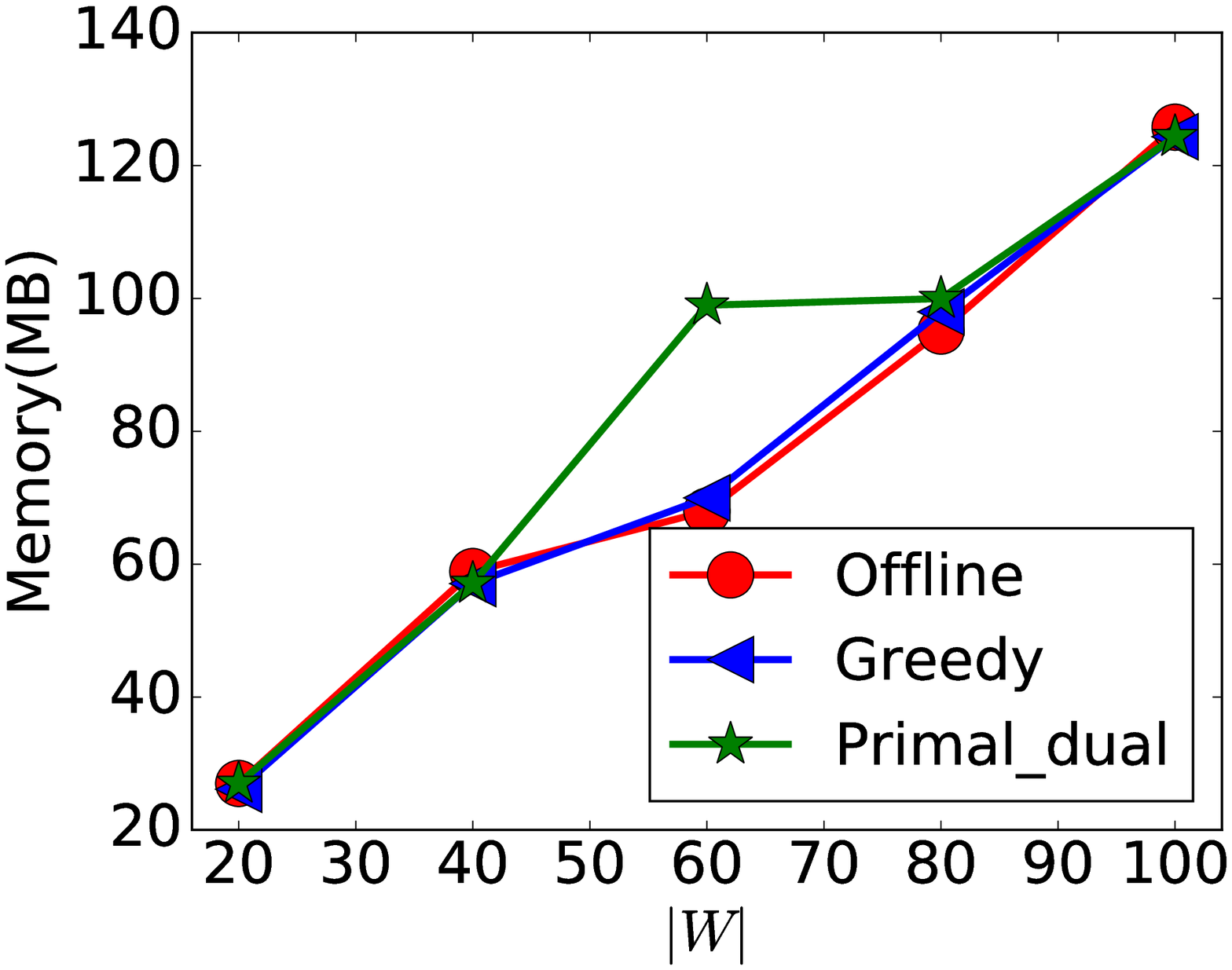}
\label{fig:w_memory}
}

\subfloat[\small{Utility of varying $|P|$}\vspace{-2ex}]{
\includegraphics[scale=0.19]{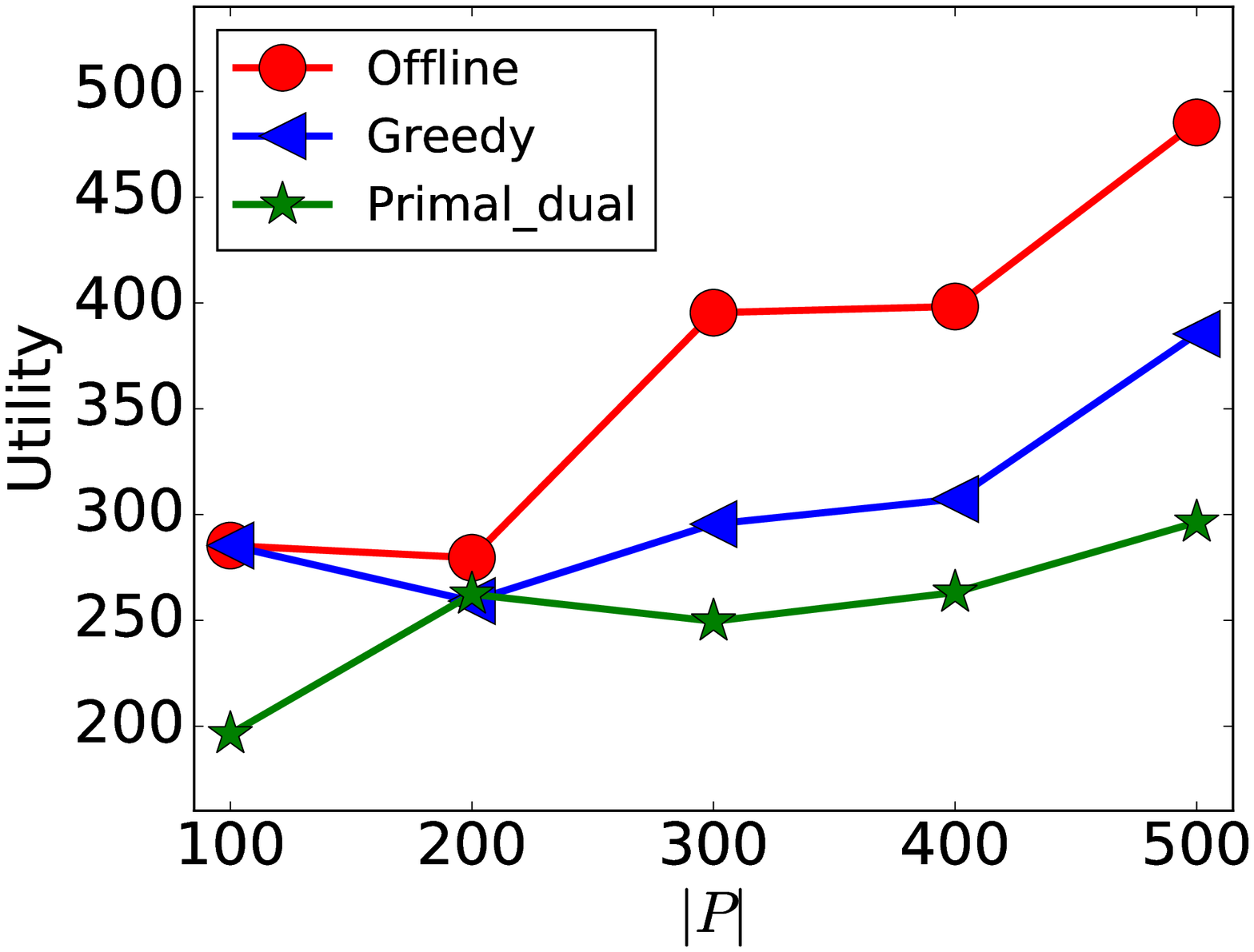}
\label{fig:p_cost}
}
~~
\subfloat[\small{Time of varying $|P|$}\vspace{-2ex}]{
\includegraphics[scale=0.19]{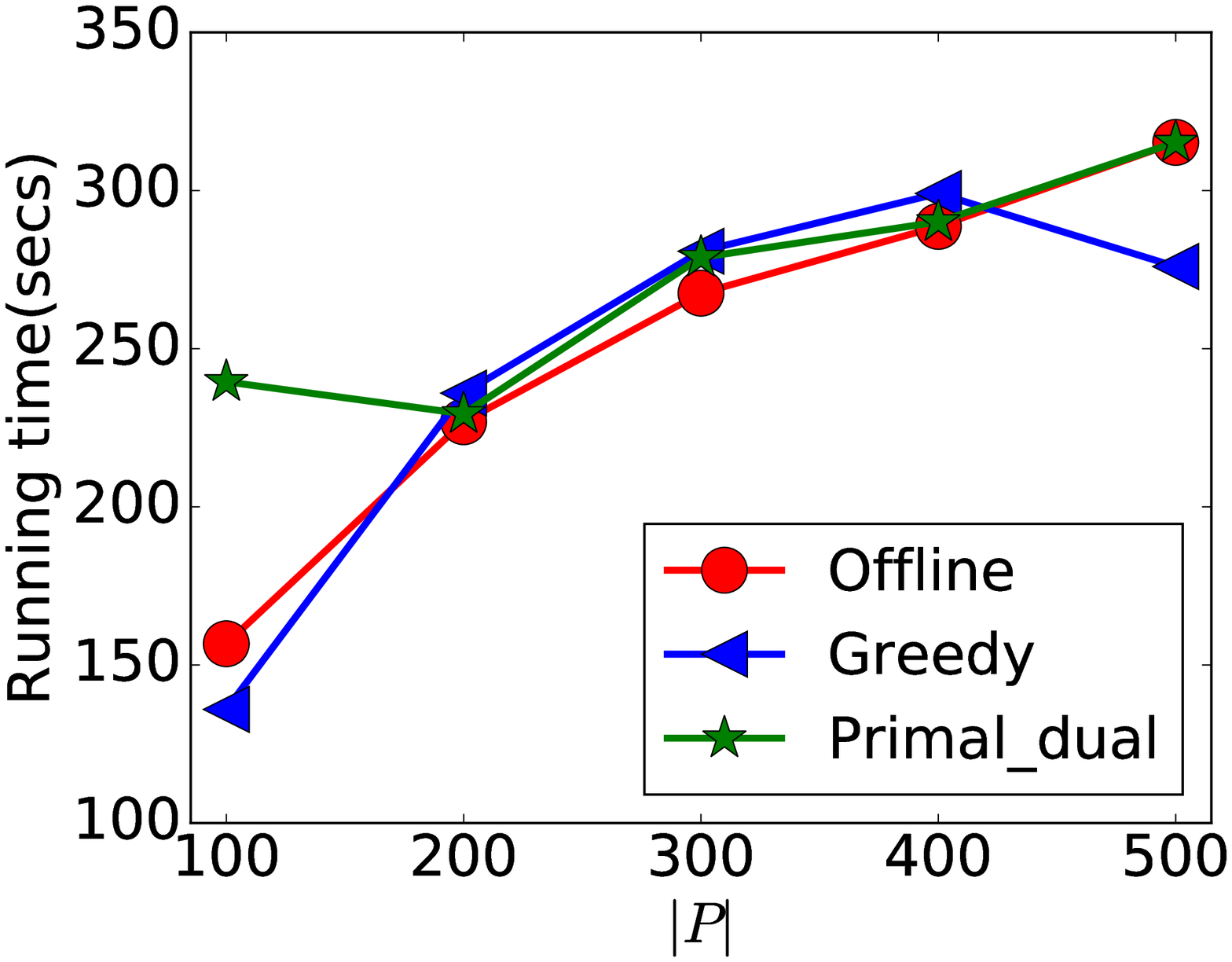}
\label{fig:p_time}
}
~~
\subfloat[\small{Memory of varying $|P|$}\vspace{-2ex}]{
\includegraphics[scale=0.19]{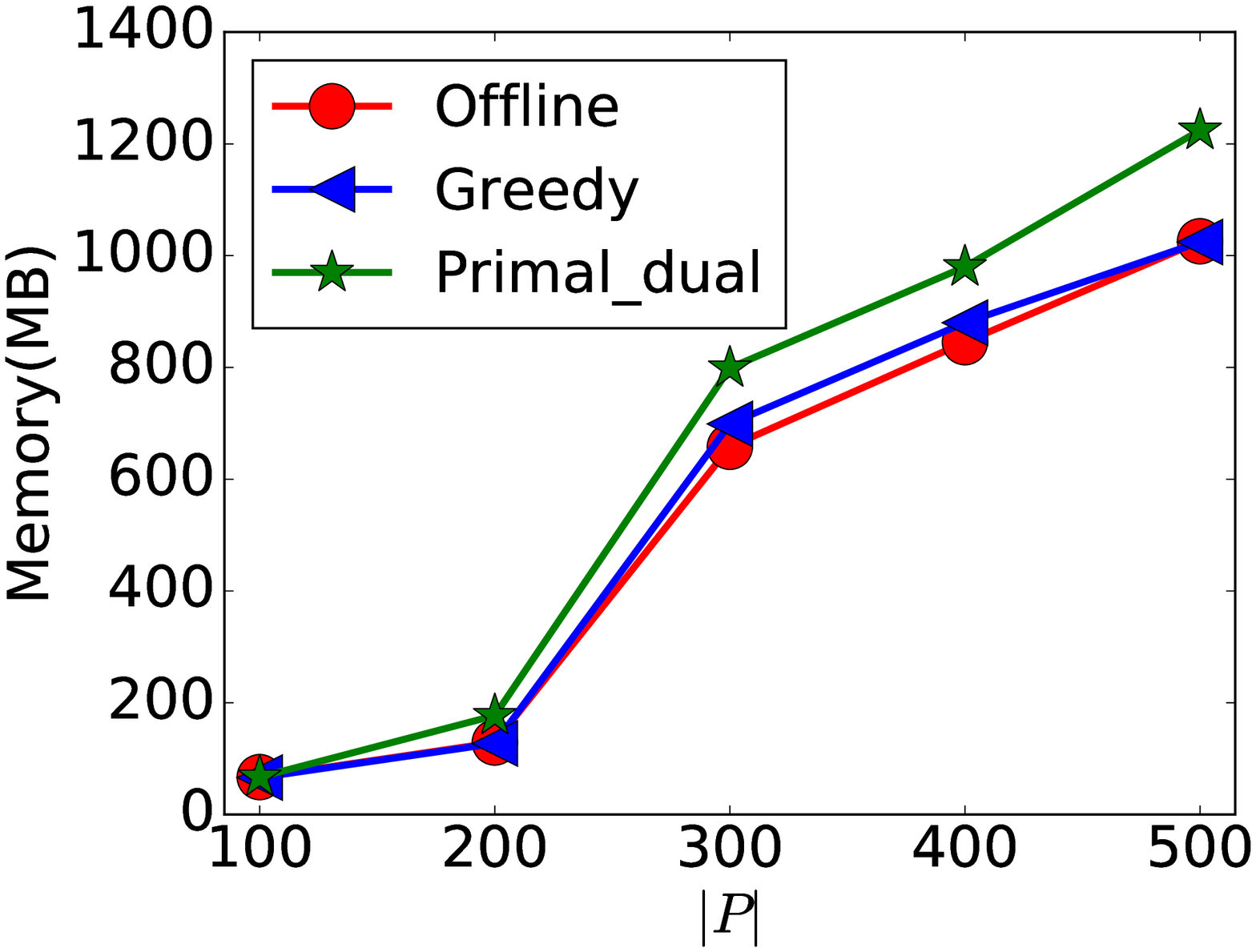}
\label{fig:p_memory}
}

\subfloat[\small{Utility of varying $|c|$}\vspace{-2ex}]{
\includegraphics[scale=0.2]{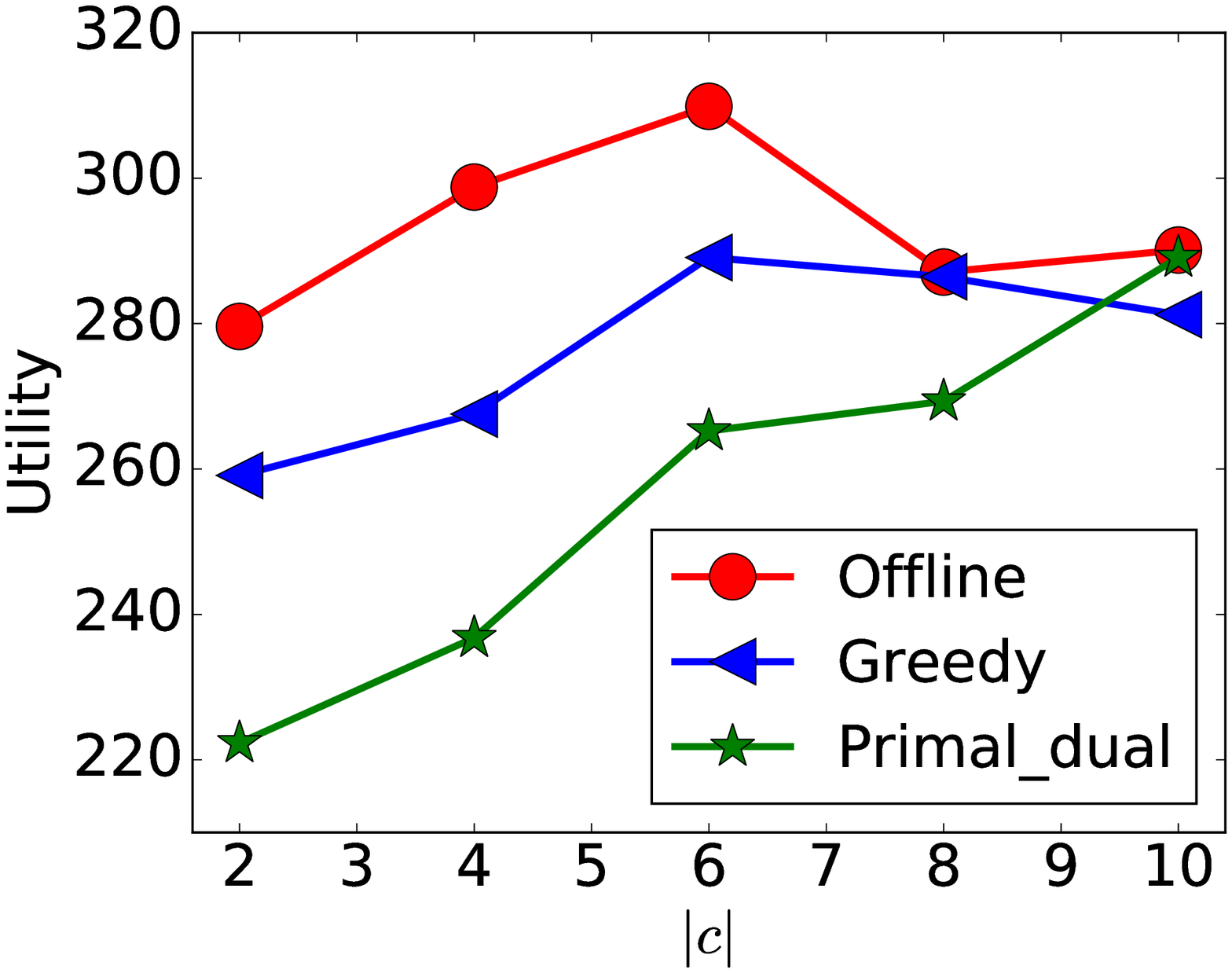}
\label{fig:c_cost}
}
~~
\subfloat[\small{Time of varying $|c|$}\vspace{-2ex}]{
\includegraphics[scale=0.19]{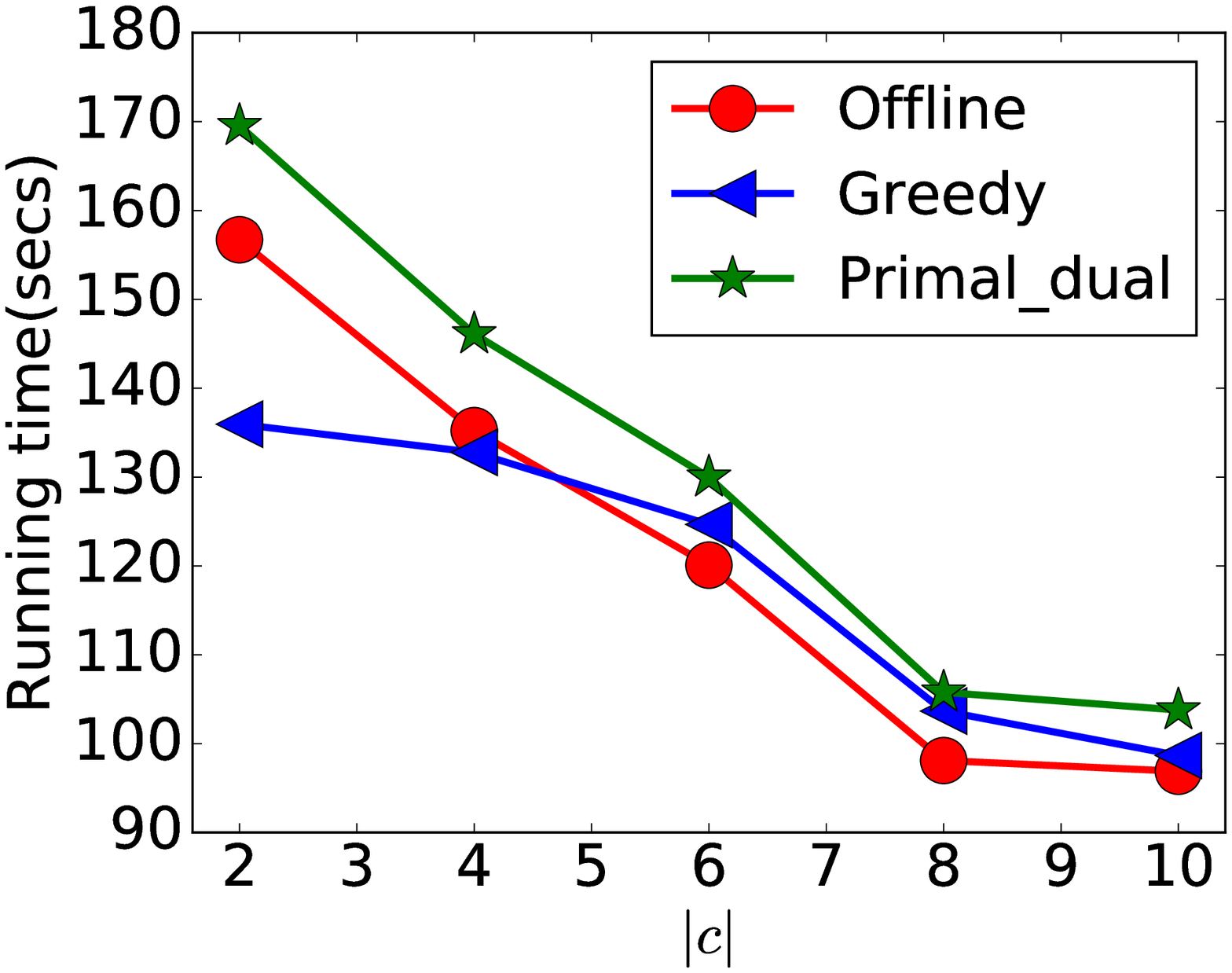}
\label{fig:c_time}
}
~~
\subfloat[\small{Memory of varying $|c|$}\vspace{-2ex}]{
\includegraphics[scale=0.19]{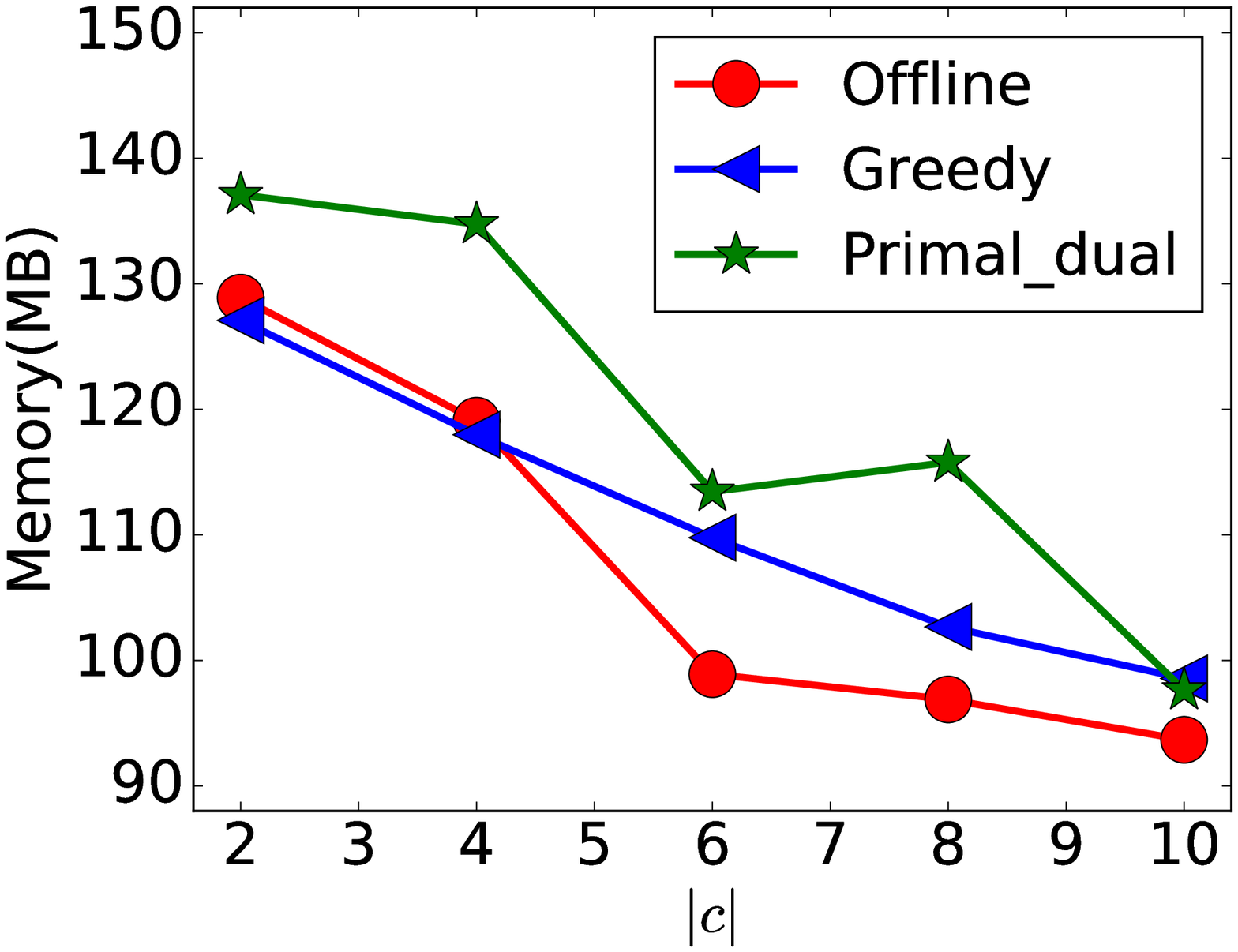}
\label{fig:c_memory}
}

\caption{Results of varying $|W|$, $|P|$, and $|c|$}
\label{fig:experiments}
\end{figure*}

{\bf Effect of $|W|$}. We first present the effects of varying $|W|$, where $|W|$ is set to \{20, 40, 60, 80, 100\}, and the number of parcels is 200.  Figs. \ref{fig:w_cost} to \ref{fig:w_memory} show the allocation utility, running time and memory cost, respectively. We have made the following observations. First, the allocation utility increases as $|W|$ increases. This is because when $|W|$ increases, the utility will be calculated among more workers. Second, the running time increases as $|W|$ varies, because as $|W|$ increases, algorithm will consume more time as the size of worker increases. Third, the memory usage increases as $|W|$ increases, which is natural, because the data becomes larger.

{\bf Effect of $|P|$}. Next, we set the number of parcels as \{100, 200, 300, 400, 500\}, and the the number of workers as 40. Figs. \ref{fig:p_cost} to \ref{fig:p_memory} show the allocation utility, running time and memory cost, respectively. We have made the following observations. First, the allocation utility increases when $|P|$ increases, because more parcels will be collected by workers. Second, the running time increases as $|P|$ increases. This is because when $|P|$ is larger and $|W|$ is fixed, one worker will take more parcels which results in more calculation time. Third, the memory cost increases as $|P|$ increases, because as $|P|$ increases, it requires more memory.

{\bf Effect of capacity $|c|$.} We set the number of parcels as 200, the number of workers as 40, and the capacity of workers as 1 to 6 random numbers, and we present the results of the allocation utility, running time and memory cost in Figs. \ref{fig:c_cost} to \ref{fig:c_memory}, respectively, when varying the capacity of the workers. As is shown, the allocation utility increase as $|c|$ increases, because an increase in $|c|$ workers will collect more parcels. Running time and memory cost decrease as worker's capacity increases.

\begin{figure*}[t]
\centering

\subfloat[\small{Utility of varying $\mu$}\vspace{-2ex}]{
\includegraphics[scale=0.19]{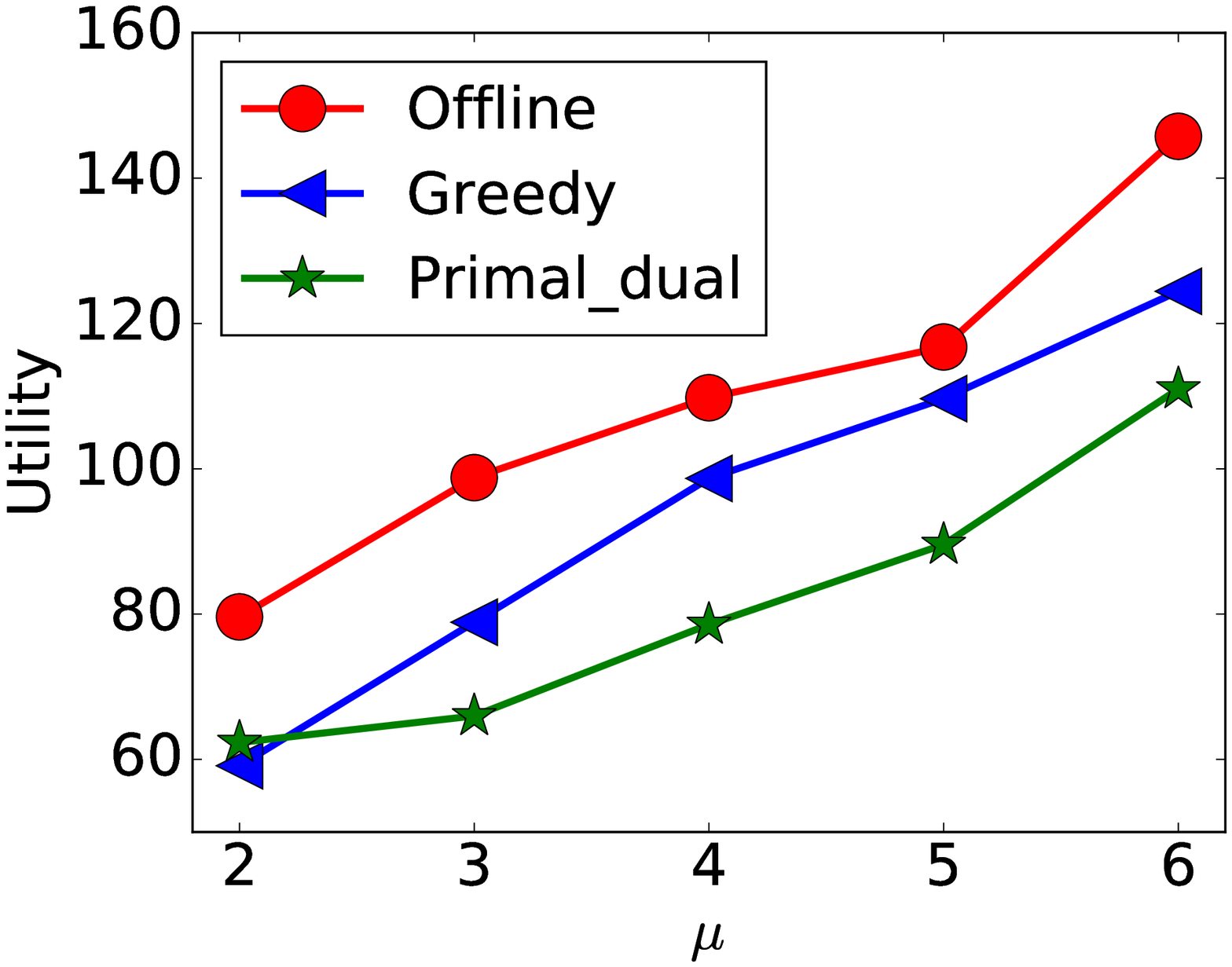}
\label{fig:mu_cost}
}
~~
\subfloat[\small{Time of varying $\mu$}\vspace{-2ex}]{
\includegraphics[scale=0.19]{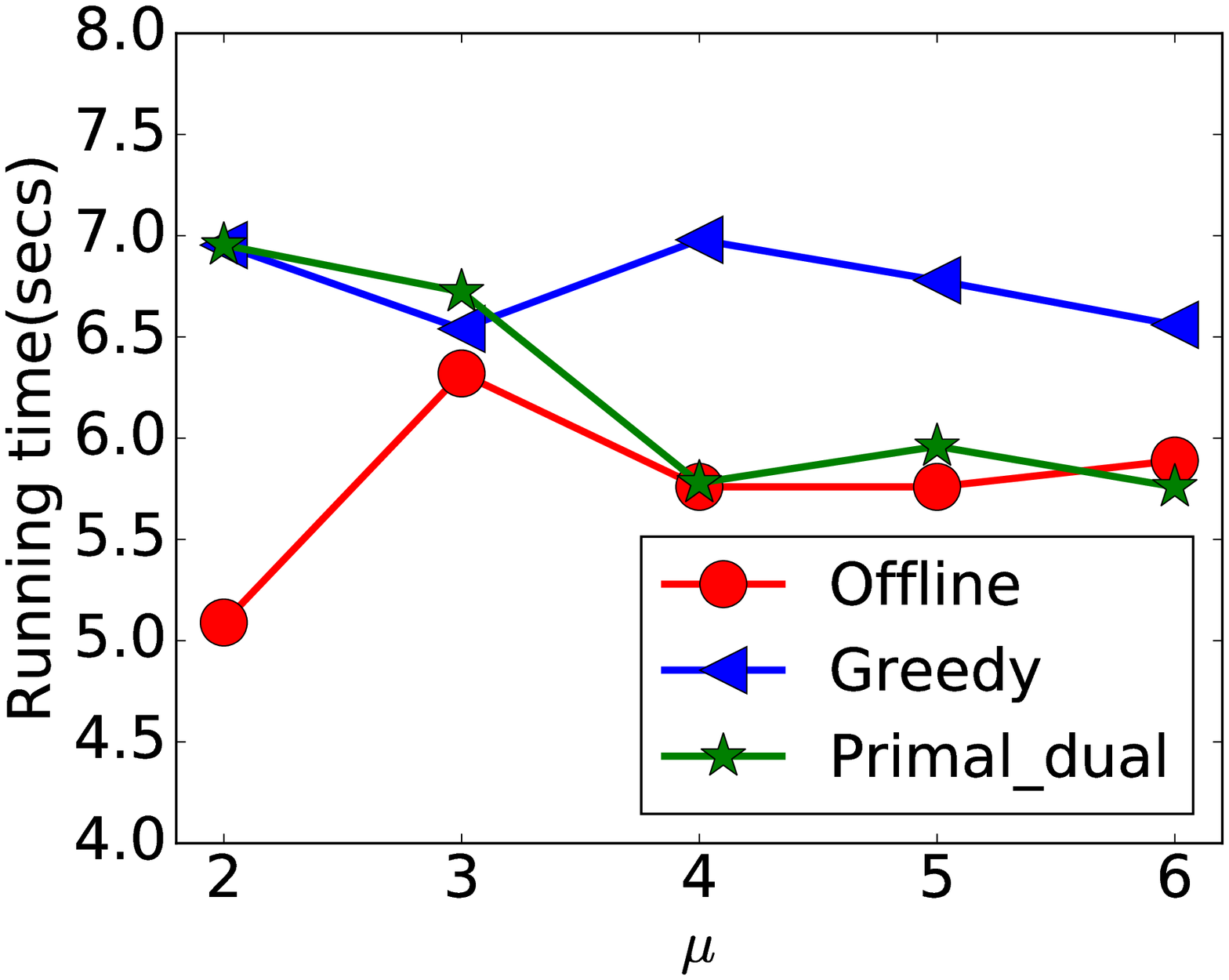}
\label{fig:mu_time}
}
~~
\subfloat[\small{Memory of varying $\mu$}\vspace{-2ex}]{
\includegraphics[scale=0.19]{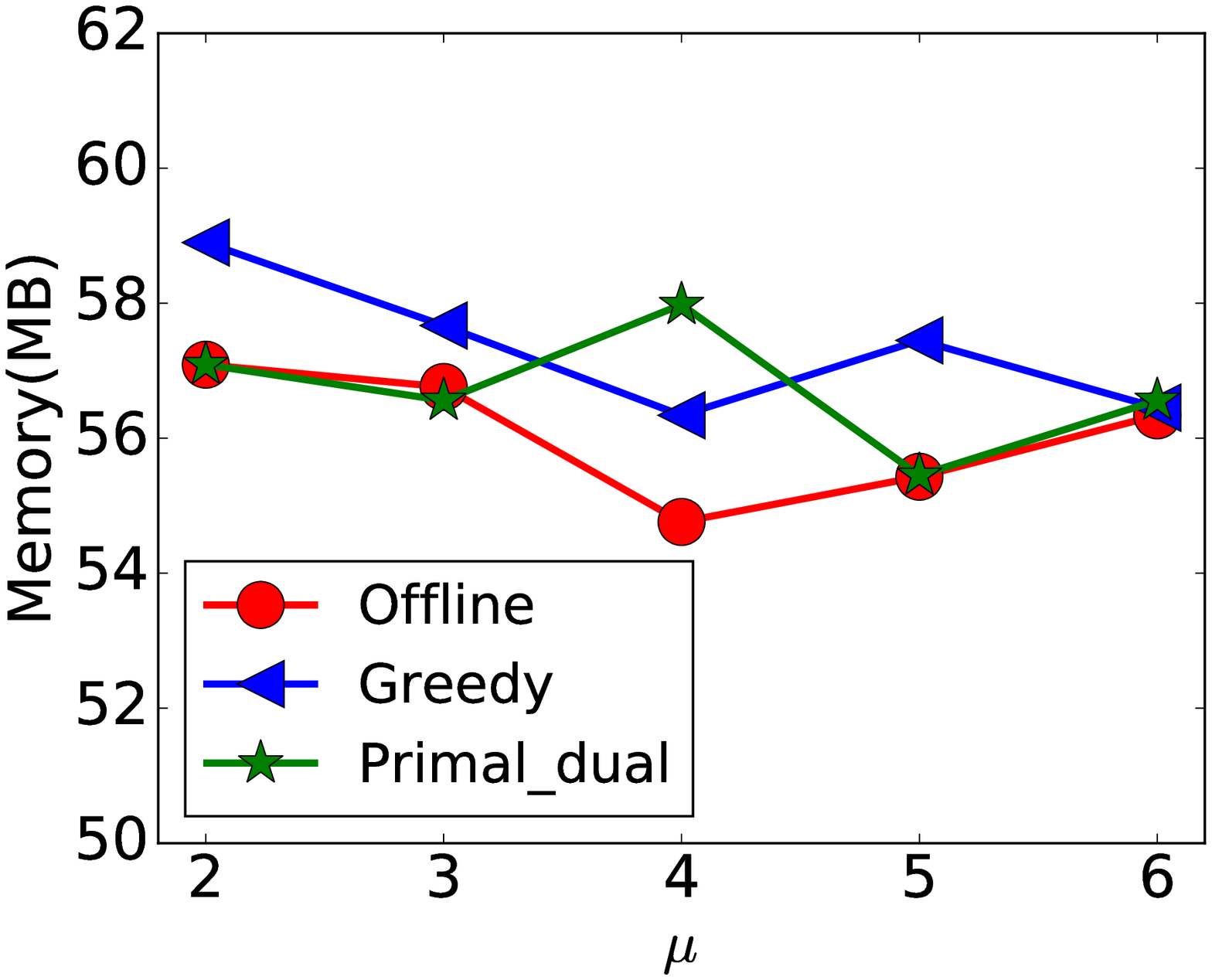}
\label{fig:mu_memory}
}

\subfloat[\small{Utility of varying $\sigma$}\vspace{-2ex}]{
\includegraphics[scale=0.19]{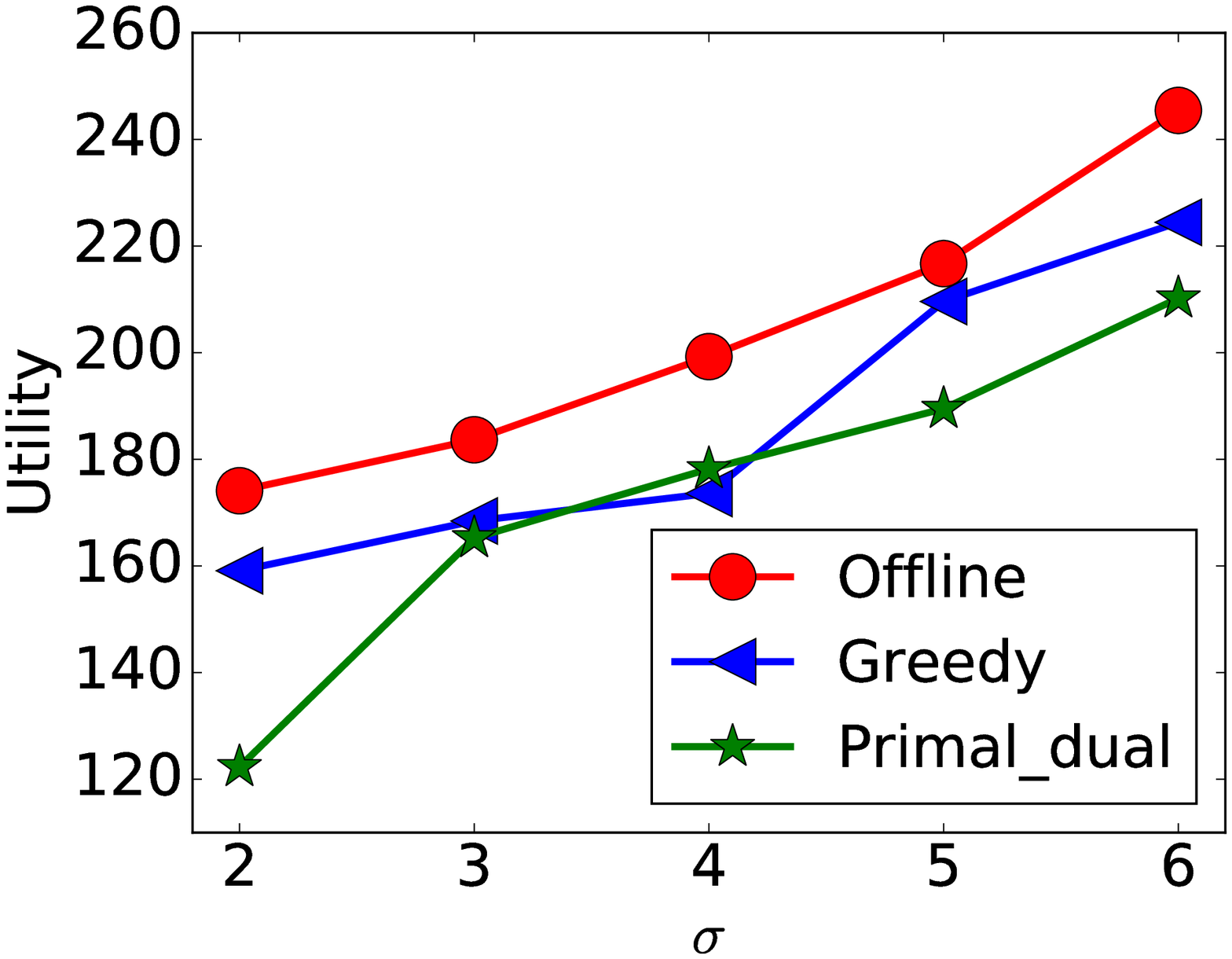}
\label{fig:sigma_cost}
}
~~
\subfloat[\small{Time of varying $\sigma$}\vspace{-2ex}]{
\includegraphics[scale=0.19]{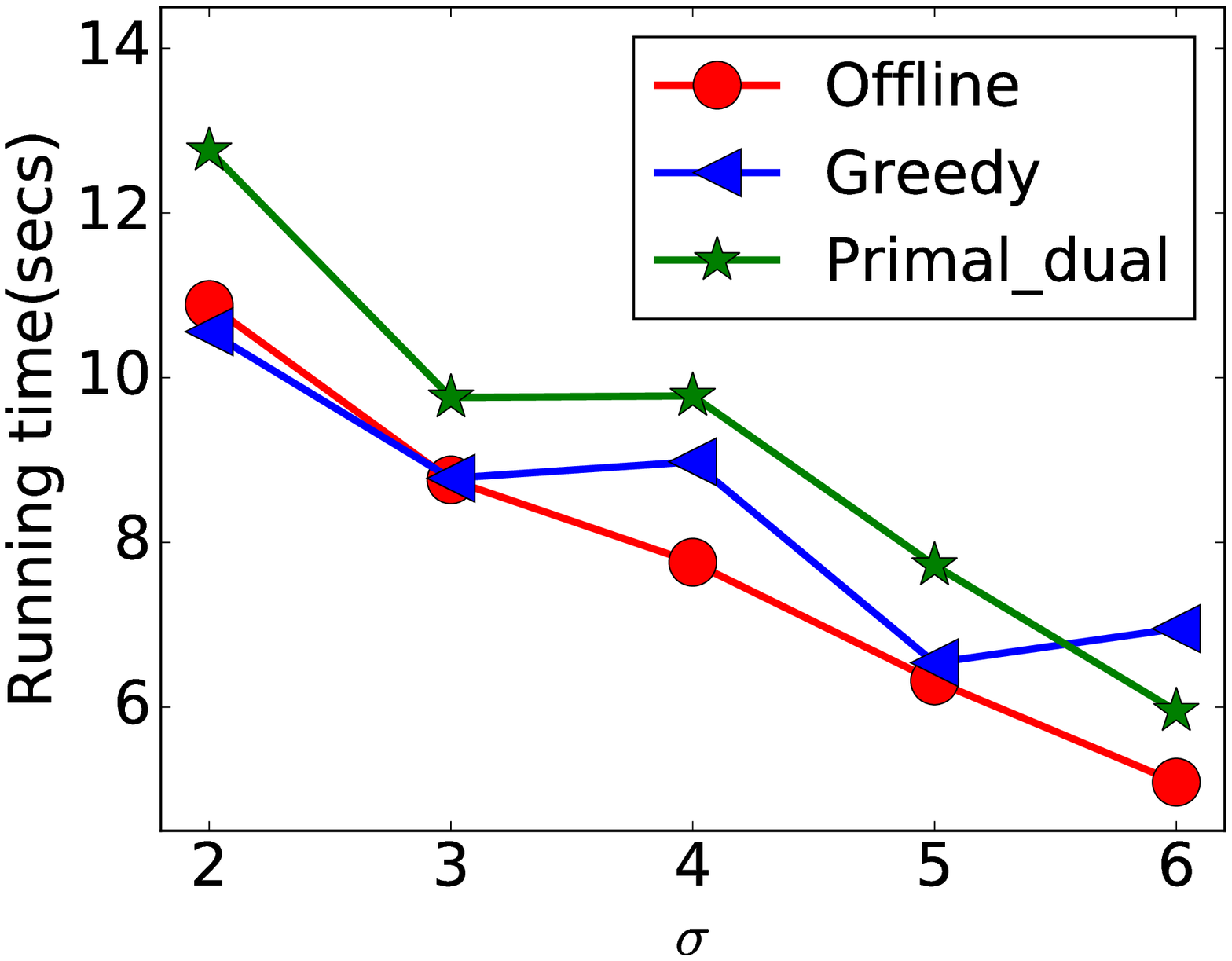}
\label{fig:sigma_time}
}
~~
\subfloat[\small{Memory of varying $\sigma$}\vspace{-2ex}]{
\includegraphics[scale=0.19]{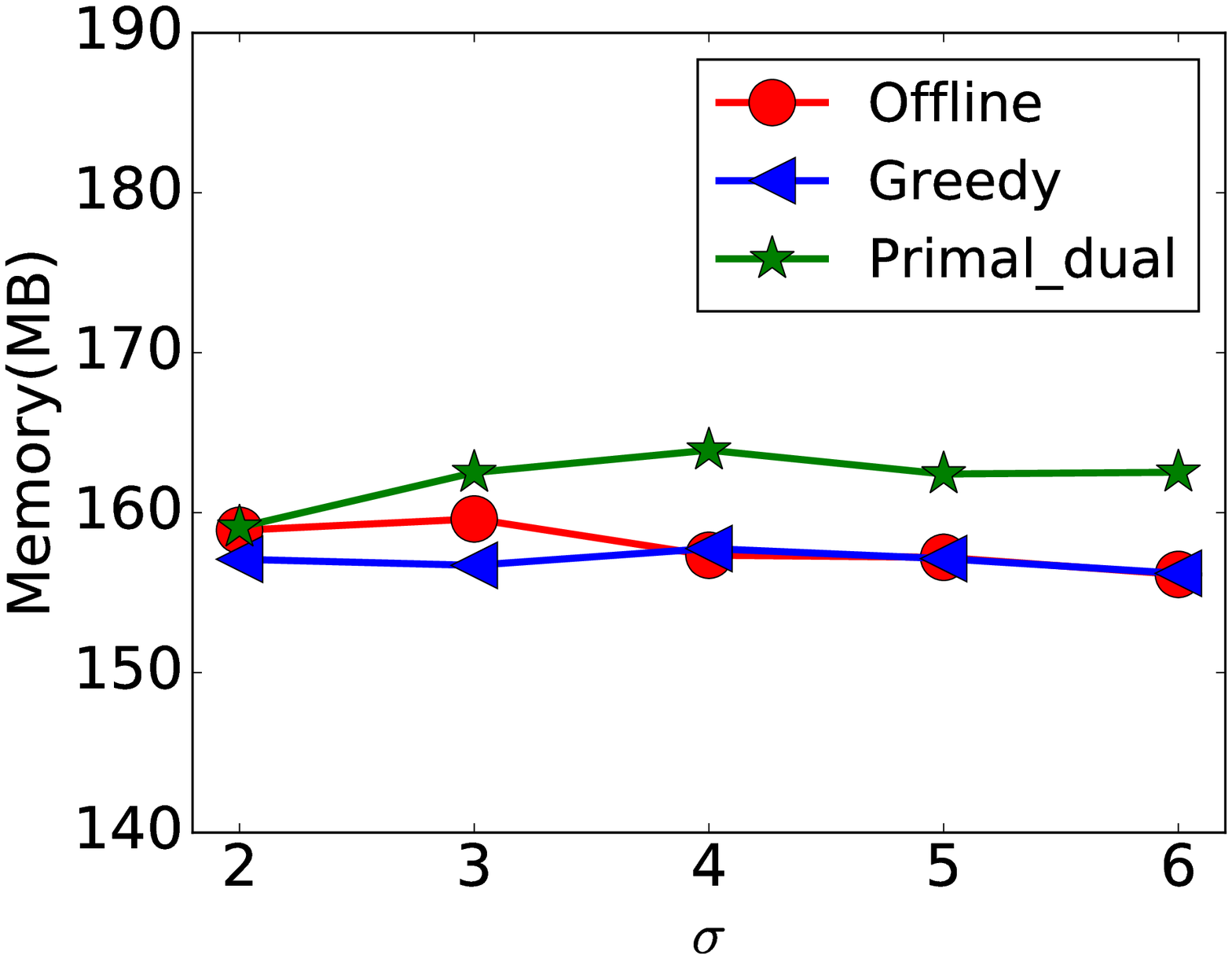}
\label{fig:sigma_memory}
}

\subfloat[\small{Utility on real datasets}\vspace{-2ex}]{
\includegraphics[scale=0.19]{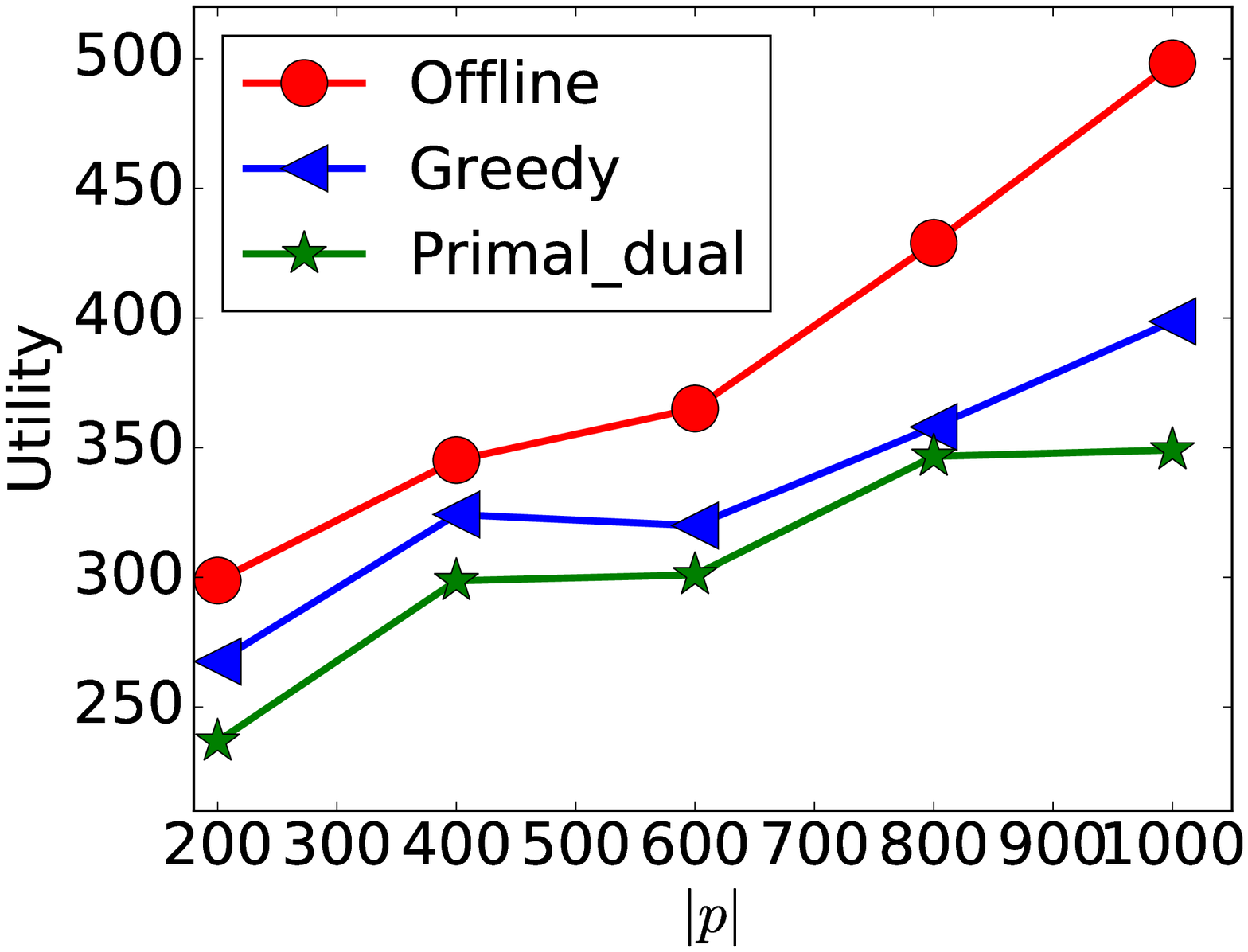}
\label{fig:r_cost}
}
~~
\subfloat[\small{Time on real datasets}\vspace{-2ex}]{
\includegraphics[scale=0.19]{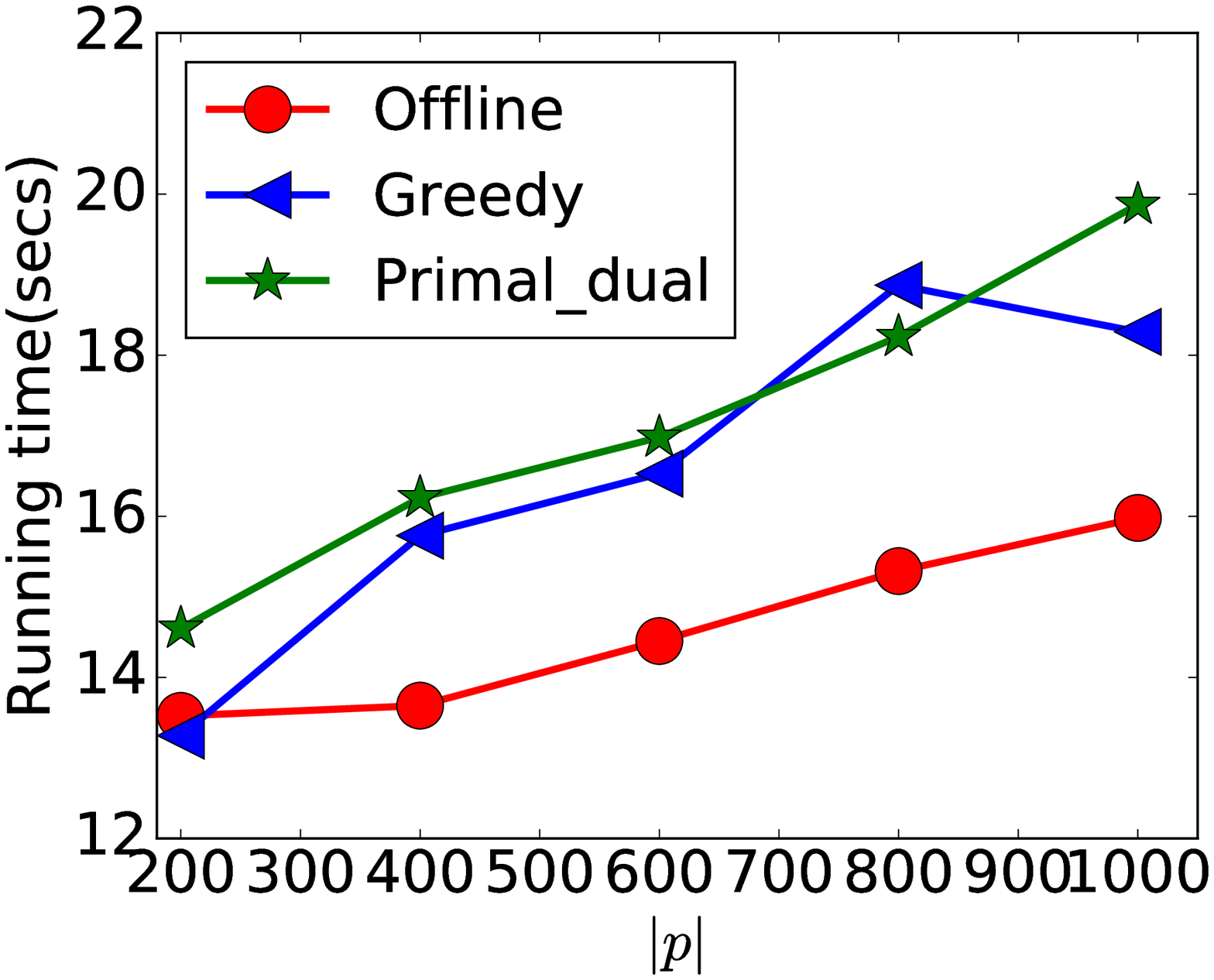}
\label{fig:r_time}
}
~~
\subfloat[\small{Memory on real dataset}\vspace{-2ex}]{
\includegraphics[scale=0.19]{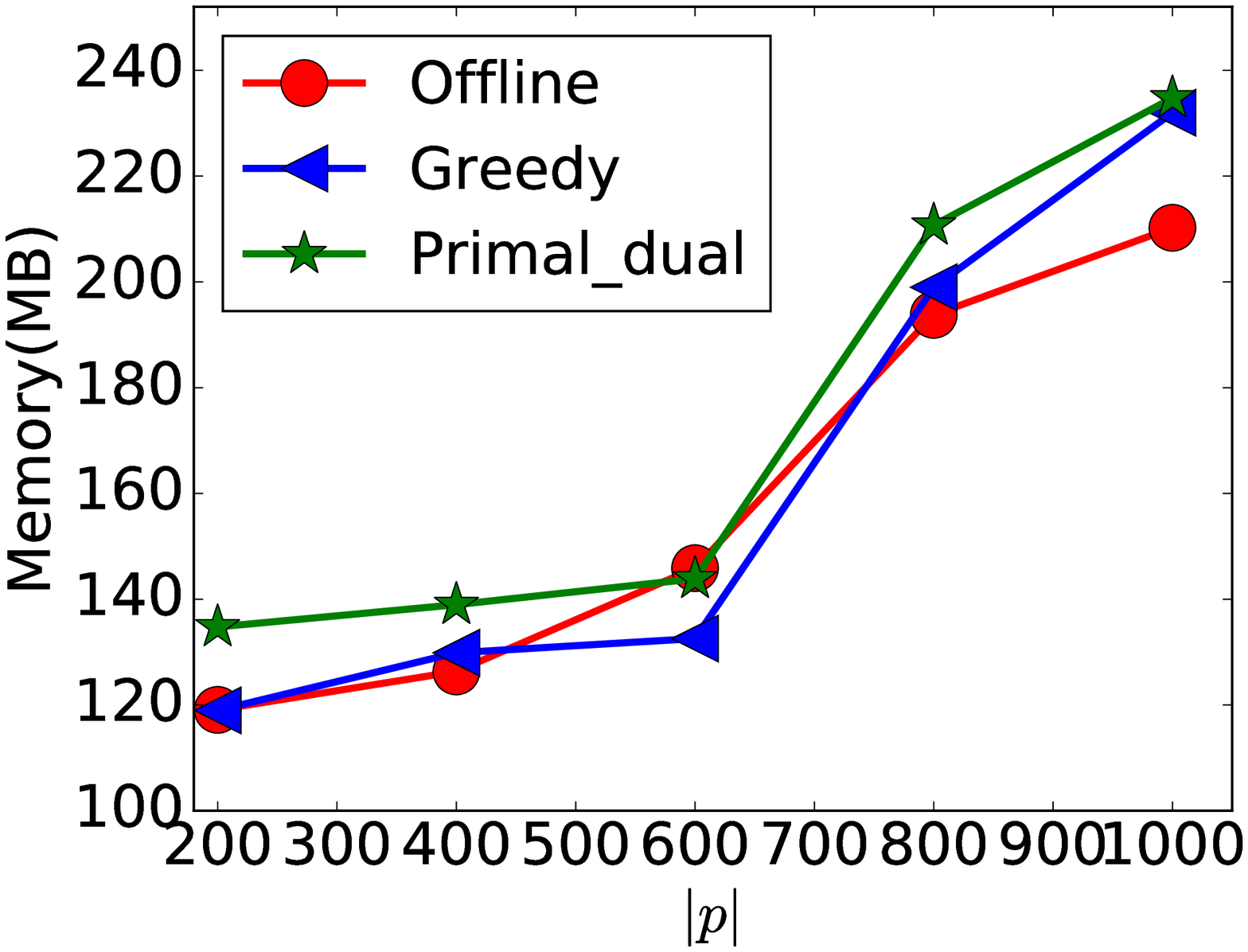}
\label{fig:r_memory}
}

\caption{Results varying working hours $|T|$ and real datasets.}
\label{fig:experiments1}
\end{figure*}

{\bf Effect of working hours $|T|$.} We set the number of parcels as 200, the number of workers as 40, and the working hours of workers following normal distribution $\mu = \{2,3,4,5,6\}, \sigma = 5$, and we present the results of the allocation utility, running time and memory cost in Figs. \ref{fig:mu_cost} to \ref{fig:mu_memory}, respectively, when varying the working hours of workers. As is shown, the allocation utility increases as $|T|$ increases since workers will have more time to collect parcels; running time and memory costs have not change much as $|T|$ increases since the number of workers and the number of parcel are fixed. In addition, we set the number of workers and the number of parcels, similar to above. And we set the working hours of workers following normal distribution $\mu = 5, \sigma = \{2,3,4,5,6\}$. Then we present the results of the allocation utility, running time and memory cost in Figs. \ref{fig:sigma_cost} to \ref{fig:sigma_memory}, respectively. As are shown, the allocation utility increase as $|T|$ increases, because an increase in $|T|$ requires more calculations. Running time decreases when varying the $\sigma$ of working hours since parcel allocation will need less worker in one day; memory cost has not change.

{\bf Real dataset.} Figs. \ref{fig:r_cost} to \ref{fig:r_memory} show the results of the allocation utility, running time and memory cost, respectively, on the real dataset \cite{zheng2009mining}. The results on this real dataset present patterns similar to the results with synthetic data.

\begin{figure*}[t]
\centering
\subfloat[\small{Utility of scalability test}\vspace{-2ex}]{
\includegraphics[scale=0.19]{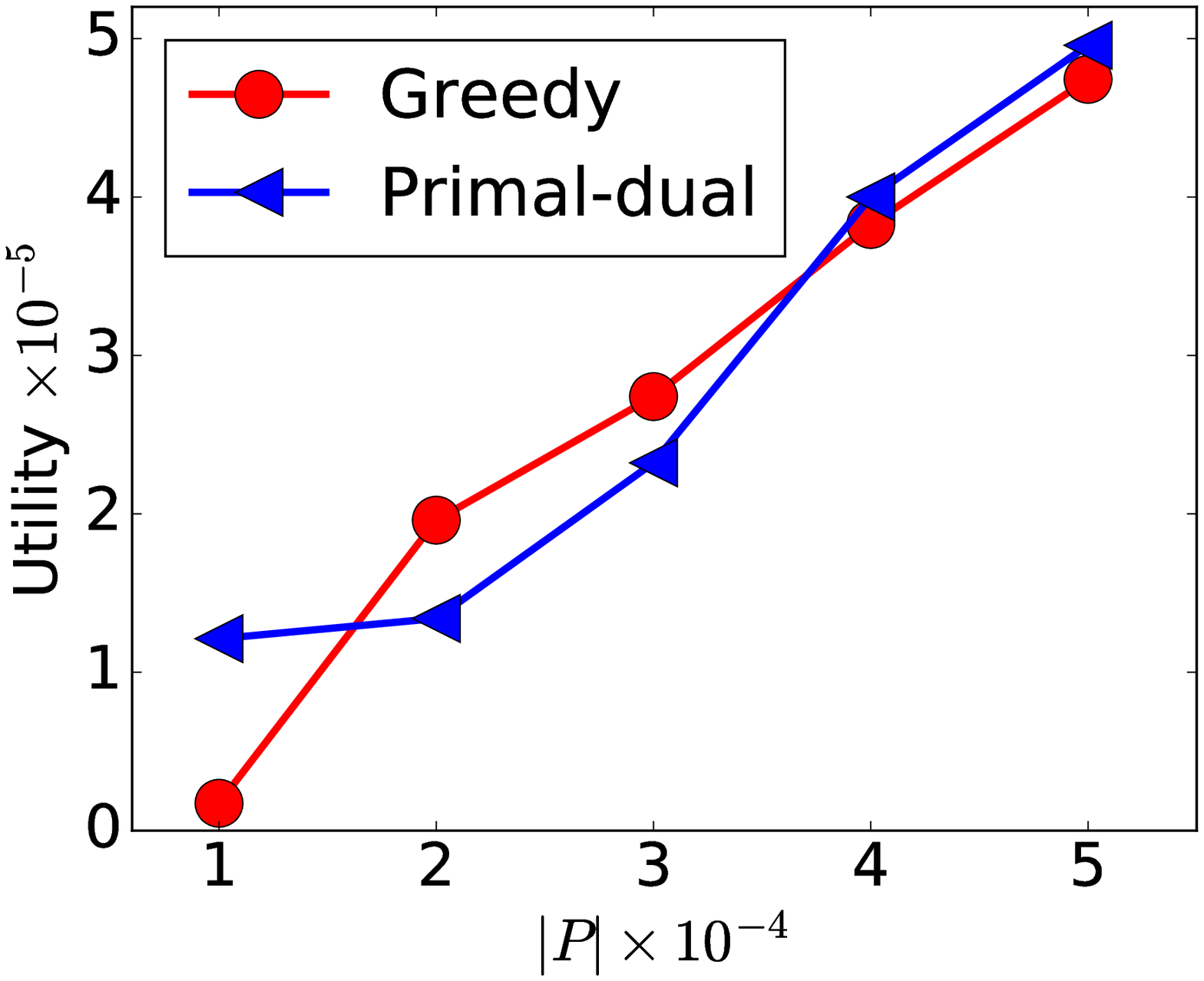}
\label{fig:s_cost}
}
~~
\subfloat[\small{Time of scalability test}\vspace{-2ex}]{
\includegraphics[scale=0.19]{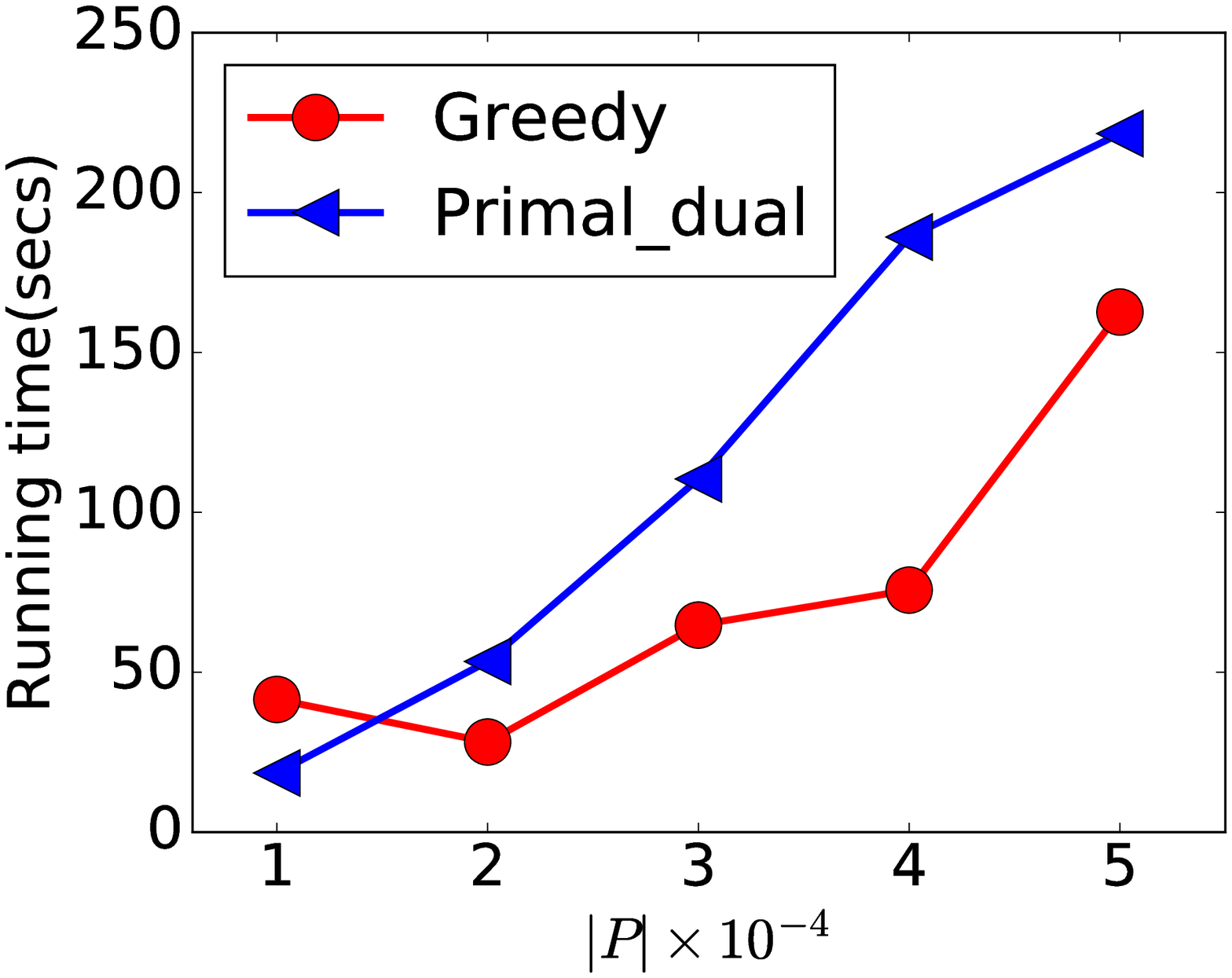}
\label{fig:s_time}
}
~~
\subfloat[\small{Memory of scalability test}\vspace{-2ex}]{
\includegraphics[scale=0.19]{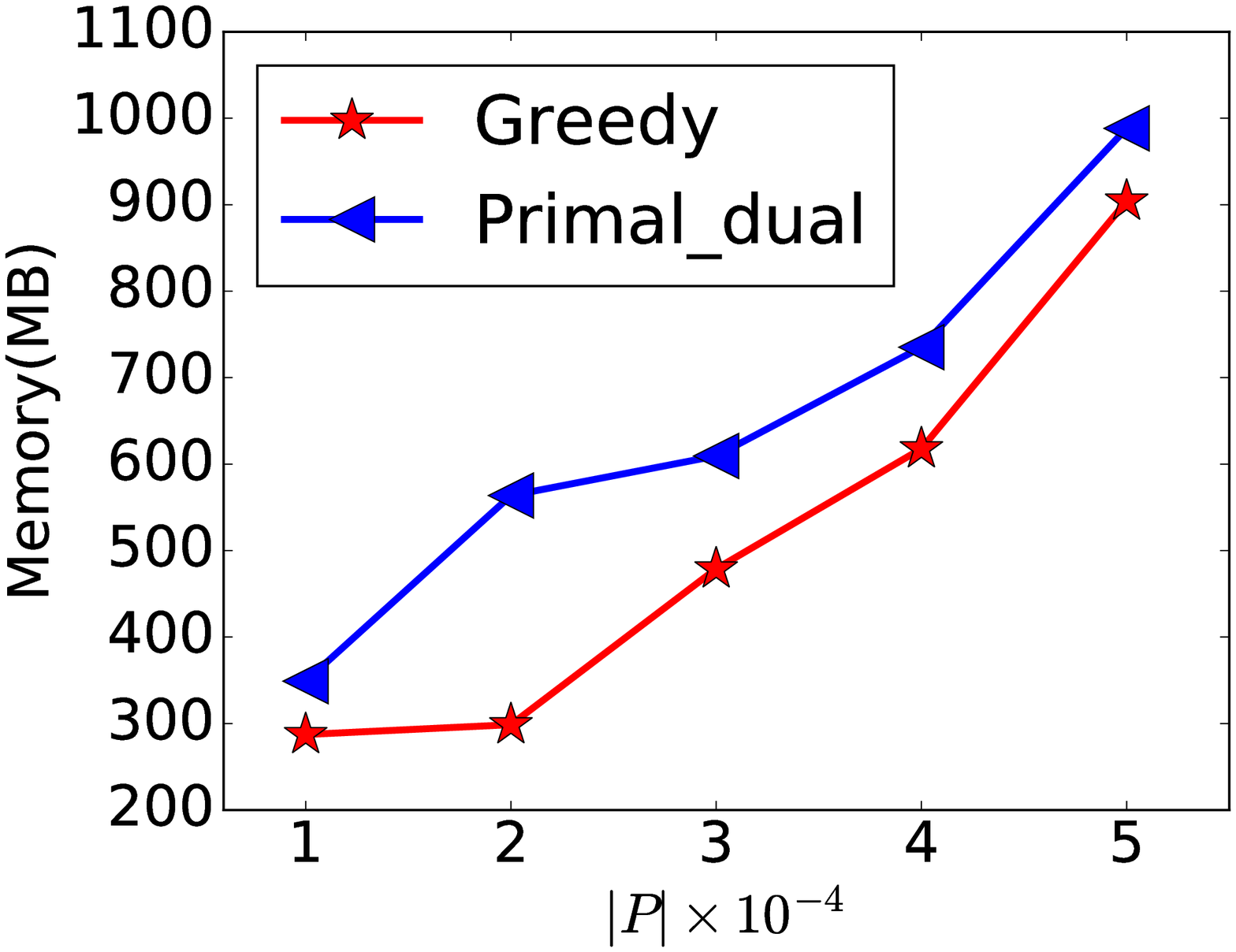}
\label{fig:s_memory}
}

\caption{Results on scalability test.}
\label{fig:experiments2}
\end{figure*}

{\bf Scalability.} Finally, we study the scalability of all the proposed algorithms. Specifically, we set $|W|$ to 200 and $|P|$ to $10k$, $20k$, $30k$, $40k$, and $50k$. The results are shown in Figs. \ref{fig:s_cost} to \ref{fig:s_memory} in terms of allocation utility, running time and memory cost, respectively. We can observe that the allocation utility, running time, and memory cost of all algorithms grow linearly with the size of the data. In addition, the result shows that all the algorithms are scalable in terms of both time and memory cost.

{\bf Summary.} The Greedy and Primal-dual algorithms are efficient for parcel allocation in last-mile delivery. The Greedy algorithm performs better than Primal-dual in allocation utility, and it is more efficient with regard to running time in more cases since Greedy do not need more time to calculate the value of dual variants.

\section{Conclusion}

In this paper, in order to solve online parcel allocation problem in last-mile delivery, we first identify a model for online parcel allocation in last-mile delivery. Then, we propose a baseline algorithm to solve this problem. We also present a primal-dual algorithm whose competitive ratio depends logarithmically on a certain parameter of the problem instance. Finally, we verify the effectiveness and efficiency of the proposed solutions through extensive experiments on both real and synthetic datasets. The experiments show promising results for our proposed algorithms.

\bibliographystyle{splncs}
\bibliography{last_mile}

\begin{thebibliography}{10}

\bibitem{2014traccs}
Chen, C., Cheng, S.F., Gunawan, A., Misra, A., Dasgupta, K., Chander, D.:
\newblock Traccs: a framework for trajectory-aware coordinated urban
  crowd-sourcing.
\newblock In: Proceedings of the Second AAAI Conference on Human Computation
  and Crowdsourcing. (2014)  30--40

\bibitem{2016_last-mile}
Wang, Y., Zhang, D., Liu, Q., Shen, F., Lee, L.H.:
\newblock Towards enhancing the last-mile delivery: An effective crowd-tasking
  model with scalable solutions.
\newblock Transportation Research Part E: Logistics and Transportation Review
  \textbf{93} (2016)  279--293

\bibitem{ahuja1993network}
Ahuja, R.K., Magnanti, T.L., Orlin, J.B.:
\newblock Network flows: theory, algorithms, and applications.
\newblock (1993)

\bibitem{kiraly2012}
Kir{\'a}ly, Z., Kov{\'a}cs, P.:
\newblock Efficient implementations of minimum-cost flow algorithms.
\newblock arXiv preprint arXiv:1207.6381 (2012)

\bibitem{tongwww16tracking}
Tong, Y., Zhang, X., Chen, L.:
\newblock Tracking frequent items over distributed probabilistic data.
\newblock World Wide Web Journal \textbf{19}(4) (2016)  579--604

\bibitem{tongjos17}
Tong, Y., Yuan, Y., Cheng, Y., Chen, L., Wang, G.:
\newblock A survey of spatiotemporal crowdsourced data management techniques.
\newblock Journal of Software (in Chinese) \textbf{28}(1) (2017)  35--58

\bibitem{zheng2015qasca}
Zheng, Y., Wang, J., Li, G., Cheng, R., Feng, J.:
\newblock Qasca: A quality-aware task assignment system for crowdsourcing
  applications.
\newblock In: Proceedings of the 2015 ACM SIGMOD International Conference on
  Management of Data, ACM (2015)  1031--1046

\bibitem{fan2015icrowd}
Fan, J., Li, G., Ooi, B.C., Tan, K.l., Feng, J.:
\newblock icrowd: An adaptive crowdsourcing framework.
\newblock In: Proceedings of the 2015 ACM SIGMOD International Conference on
  Management of Data, ACM (2015)  1015--1030

\bibitem{2012_online_matching}
Mehta, A.:
\newblock Online matching and ad allocation.
\newblock Theoretical Computer Science \textbf{8}(4) (2012)  265--368

\bibitem{ICDE16_Tong}
Tong, Y., She, J., Ding, B., Wang, L., Chen, L.:
\newblock Online mobile micro-task allocation in spatial crowdsourcing.
\newblock In: Proceedings of the 32nd IEEE International Conference on Data
  Engineering. (2016)  49--60

\bibitem{tongpvldb16}
Tong, Y., She, J., Ding, B., Chen, L., Wo, T., Xu, K.:
\newblock Online minimum matching in real-time spatial data: Experiments and
  analysis.
\newblock Proceedings of the VLDB Endowment \textbf{9}(12) (2016)  1053--1064

\bibitem{ho2012online}
Ho, C.J., Vaughan, J.W.:
\newblock Online task assignment in crowdsourcing markets.
\newblock In: AAAI. Volume~12. (2012)  45--51

\bibitem{zhangpvldb14}
Zhang, C., Tong, Y., Chen, L.:
\newblock Where to: Crowd-aided path selection.
\newblock Proceedings of the VLDB Endowment \textbf{7}(11) (2014)  2005--2016

\bibitem{shahabi2012}
Kazemi, L., Shahabi, C.:
\newblock Geocrowd: enabling query answering with spatial crowdsourcing.
\newblock In: Proceedings of the 20th International Conference on Advances in
  Geographic Information Systems. (2012)  189--198

\bibitem{he2014toward}
He, S., Shin, D.H., Zhang, J., Chen, J.:
\newblock Toward optimal allocation of location dependent tasks in
  crowdsensing.
\newblock In: INFOCOM, 2014 Proceedings IEEE, IEEE (2014)  745--753

\bibitem{primal-dual2009}
Buchbinder, N., Seffi)~Naor, J.:
\newblock The design of competitive online algorithms via a primal: Dual
  approach.
\newblock Foundations \& Trends® in Theoretical Computer Science
  \textbf{3}(2–3) (2009)  93--263

\bibitem{2005online}
Buchbinder, N., Naor, J.:
\newblock Online primal-dual algorithms for covering and packing problems.
\newblock In: European Symposium on Algorithms, Springer (2005)  689--701

\bibitem{2013unified}
Niazadeh, R., Kleinberg, R.D.:
\newblock A unified approach to online allocation algorithms via randomized
  dual fitting.
\newblock arXiv preprint arXiv:1308.5444 (2013)

\bibitem{primal-dual}
Lahaie, S.:
\newblock How to take the dual of a linear program.
\newblock Columbia University, New York (2008)

\bibitem{zheng2009mining}
Zheng, Y., Zhang, L., Xie, X., Ma, W.Y.:
\newblock Mining interesting locations and travel sequences from gps
  trajectories.
\newblock In: Proceedings of the 18th International Conference on World Wide
  Web, ACM (2009)  791--800

\end{thebibliography}

\end{document}